\documentclass[11pt]{article}
\usepackage[T1]{fontenc}
\usepackage[margin=1in]{geometry}
\usepackage{hyperref,cite,microtype}
\usepackage{graphicx,amssymb,amsmath,amsthm}
\usepackage{aliascnt}

\newcommand{\NNG}{\ensuremath{\mathcal{N}}}

\newtheorem{theorem}{Theorem}

\newaliascnt{lemma}{theorem}
\newtheorem{lemma}[lemma]{Lemma}
\aliascntresetthe{lemma}

\newaliascnt{corollary}{theorem}
\newtheorem{corollary}[corollary]{Corollary}
\aliascntresetthe{corollary}

\theoremstyle{definition}
\newaliascnt{definition}{theorem}
\newtheorem{definition}[definition]{Definition}
\aliascntresetthe{definition}

\title{Maximizing the Sum of Radii of Disjoint Balls or Disks}

\author{David Eppstein\thanks{Computer Science Department,
        University of California, Irvine; Irvine, CA, 92697, USA; {\tt eppstein@uci.edu}. Research supported in part by NSF grants  CCF-1228639, CCF-1618301, and CCF-1616248.}}

\begin{document}
\thispagestyle{empty}
\maketitle

\begin{abstract}
Finding nonoverlapping balls with given centers in any metric space, maximizing the sum of radii of the balls, can be expressed as a linear program. Its dual linear program expresses the problem of finding a minimum-weight set of cycles (allowing 2-cycles) covering all vertices in a complete geometric graph. For points in a Euclidean space of any finite dimension~$d$, with any convex distance function on this space, this graph can be replaced by a sparse subgraph obeying a separator theorem. This graph structure leads to an algorithm for finding the optimum set of balls in time $O(n^{2-1/d})$, improving the $O(n^3)$ time of a naive cycle cover algorithm. As a subroutine, we provide an algorithm for weighted bipartite matching in graphs with separators, which speeds up the best previous algorithm for this problem on planar bipartite graphs from $O(n^{3/2}\log n)$ to $O(n^{3/2})$ time. We also show how to constrain the balls to all have radius at least a given threshold value, and how to apply our radius-sum optimization algorithms to the problem of embedding a finite metric space into a star metric minimizing the average distance to the hub.
\end{abstract}

\section{Introduction}

In this paper we consider the following problem. We are given a set of points in the plane, and must choose a radius for each point, so that the disks with the given radii do not overlap and the sum of radii is as large as possible.
The same problem can be generalized to arbitrary metric spaces, with some care about definitions to allow us to work with finite metric spaces (such as a space whose only points are the given centers).
Although we believe these problems to be natural and interesting on their own merits, we provide two motivating applications, in map labeling and in optimal embedding of metric spaces into star metrics.

To formally define the problems we study, we use the following definition.

\begin{definition}
We define a \emph{ball}, in a given metric space, to be a pair $(c,r)$ where $c$ is a point in the metric space (the center of the ball), and $r$ is a non-negative real number (the radius of the ball). A ball $(c,r)$ contains a point $p$ if the distance from $c$ to $p$ is less than $r$.
Two balls in an arbitrary metric space \emph{overlap} when  their sum of radii exceeds the center distance, and are \emph{nonoverlapping} otherwise.
If the sum of radii equals the center distance, we say that the two balls \emph{touch}.
\end{definition}

For a \emph{geodesic} metric space such as a Euclidean space, in which every two points can be joined by an isometrically embedded line segment, this definition coincides with the usual intuition: balls are distinct when they contain distinct sets of points, and overlap when they have a point of intersection. However, these definitions allow us to work more conveniently with non-geodesic metric spaces, such as a finite metric space consisting only of the ball centers. In such spaces, distinct balls may contain the same sets of points, and overlapping balls may not have any points that they both contain.

Then in any metric space, finding a set of nonoverlapping metric balls with $n$ given center points $p_i$, maximizing the sum of non-negative radii $r_i$, can be expressed as a linear program. The objective is to maximize the linear function $\sum r_i$, subject to constraints that each pair of balls remain nonoverlapping, i.e. that $r_i+r_j\le d(p_i,p_j)$. This linear program has two variables per constraint, a well-studied special case of linear programming. But although strongly polynomial algorithms for feasibility with two variables per constraint are known~\cite{CohMeg-SICOMP-94,HocNao-SICOMP-94}, they do not extend to optimization, and their running time is higher than might be desired.
Therefore, it remains of interest to find a purely combinatorial algorithm for the problem, with as low a running time as possible.

In this paper, we provide such a combinatorial algorithm, running in cubic time for general metrics and subquadratic time for low-dimensional Euclidean spaces.

\begin{figure}[t]
\centering\includegraphics[scale=0.45]{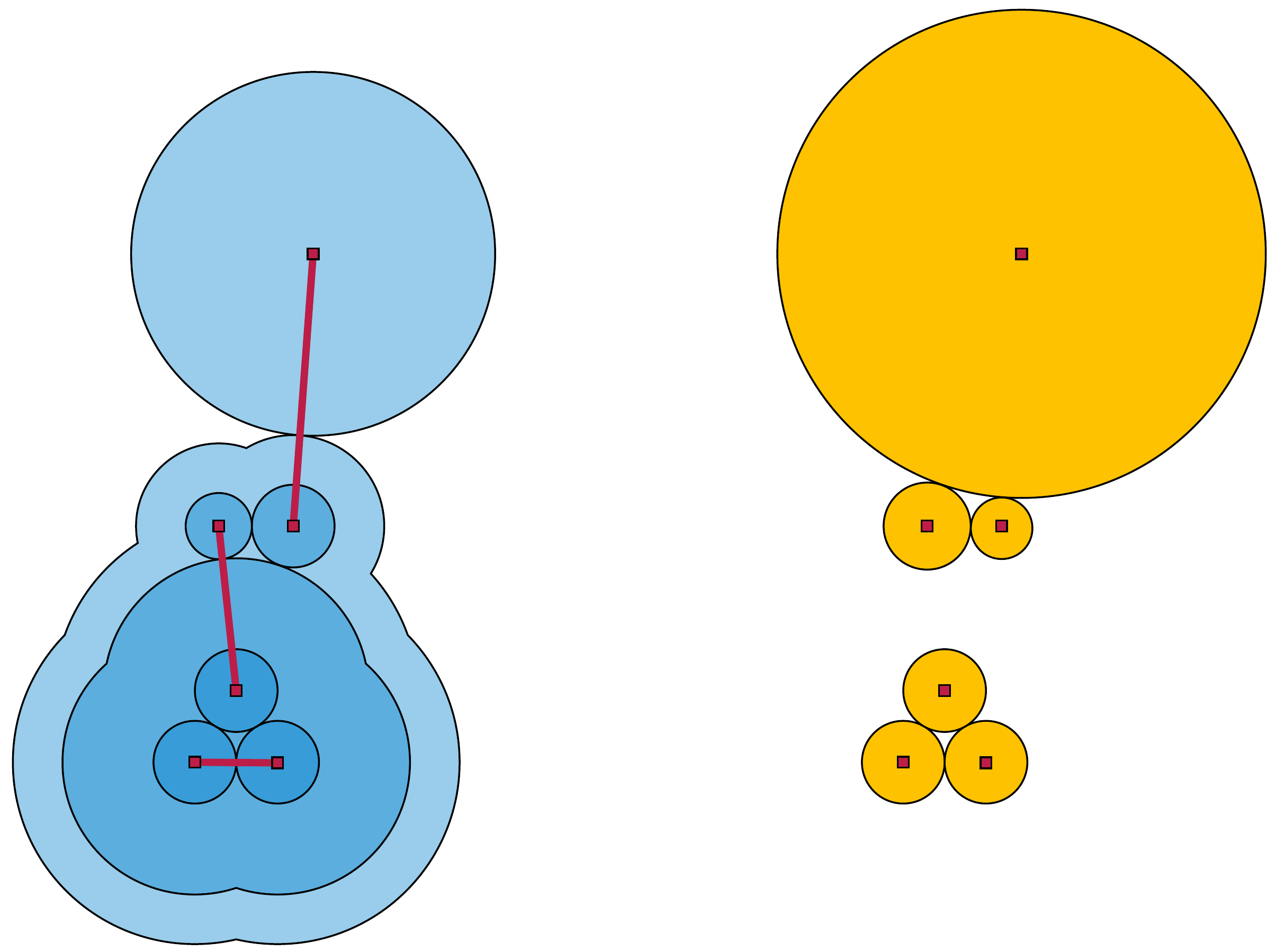}
\caption{The moat-growing method for Euclidean matching (left) versus the set of disks with maximum sum of radii centered at the same set of points (right).}
\label{fig:moats}
\end{figure}

The problem that we study here should be distinguished from the classical ``moat-growing'' method for solving Euclidean minimum weight matching problems, in which we grow disks around points until they touch but then, around any subset of an odd number of points with touching regions, continue growing a ``moat''. This method was illustrated, for instance, on the cover of a combinatorial optimization book by Cook et al.~\cite{CooCunPul-CO-98}. In contrast, our method has only disks, with no moats, and when it finds a subsystem of an odd number of points with touching disks, those disks are frozen in place, with no possibility of future growth (\autoref{lem:odd-rad}).
Although similar in visual appearance to moat-growing (\autoref{fig:moats}), the solution to our problem can have a substantially smaller value than the solution to the Euclidean minimum weight matching problem. The radius-sum optimization problem that we consider here was formulated earlier by J\"unger and Pulleyblank~\cite{JunPul-Algo-95}, as a type of geometric dual problem to Euclidean minimum weight matching, but they quickly dismissed it as inaccurate before formulating tighter moat-growing-based heuristics for matching, and they did not provide an algorithm for finding the optimal radii.

\subsection{New results}
We prove the following results:
\begin{itemize}
\item Finding metric balls with maximum sum of radii is equivalent under linear programming duality to finding a minimum-length set of cycles (allowing $2$-cycles) that cover all vertices of the complete geometric graph on the given centers. The maximum sum of radii equals half of the minimum total cycle length. By reducing cycle covers to weighted matching, the optimal set of balls in any metric space can be constructed in cubic time.
\item For points in Euclidean spaces of bounded dimension $d$, the edges of the optimal cycle cover can be found in a subgraph of the complete geometric graph, the intersection graph of nearest neighbor balls of the points. This graph is sparse, having $O(n)$ edges with a constant of proportionality depending singly-exponentially in the dimension. Moreover, it obeys a separator theorem with separators of size $O(n^{1-1/d})$, and the graph and its separator decomposition can be constructed in time $O(n\log n)$.
\item A separator-based divide and conquer weighted matching algorithm can find an optimal cycle cover and set of nonoverlapping balls in time $O(n^{2-1/d})$. In particular,  we can find an optimal set of disks in the plane in time $O(n^{3/2})$.
\item As a subroutine, we use a new weighted bipartite matching algorithm for families of graphs that obey a separator theorem, which runs in time $O(n^{3/2})$ for planar bipartite graphs, improving a previous bound of $O(n^{3/2}\log n)$ for this case.
\item We show how to solve radius-sum optimization problems in which all disks or balls have a minimum radius greater than zero, by transforming these problems into unconstrained radius-sum optimization problems on a different metric space.
\end{itemize}
To avoid issues of numerical precision we consider only strongly-polynomial-time algorithms, in a model of computation in which distance computation, arithmetic, and comparisons take constant time per operation.

\subsection{Related research}

Along with the map labeling research discussed below, researchers have studied other types of geometric optimization problems in which the optimization criterion is a sum of radii of balls or circles.
These include finding a set of $k$ balls with given centers, drawn from a larger set of $n$ points, that cover all points and minimize the sum of radii of the balls~\cite{GibKanKro-Algo-10}, finding a connected set of disks in the plane with given centers that minimize the sum of radii~\cite{ChaFekHof-WADS-11}, and finding both a collection of disks centered at a subset of input points that covers all input points, and a tour connecting the disk centers, minimizing a combination of disk radii and tour length~\cite{AltArkBro-SoCG-06}.

\section{Applications}
\subsection{Map labeling}
\label{sec:label}

\emph{Map labeling} is an algorithmic problem in which one must place non-overlapping text labels on maps or other  visualizations.
Researchers of map labeling have studied various problems of assigning shapes of maximum size to a given set of points in the plane~\cite{DodMarMir-SODA-97,WagWol-CGTA-97,JiaBerQin-ISAAC-04,Str-IJCGA-01}. One simple case of this problem, finding circles centered at the given points that maximize the minimum radius, can be solved by finding the closest pair of points and setting all radii equal to half of this pair's distance. However, this measure of solution quality penalizes the label sizes even for points far away from the closest pair, where larger radii could be used without overlap (\autoref{fig:compare-opt-criteria}, left). On the other hand, an $L_2$ measure of solution quality (maximizing the sum of disk areas) would also be unsatisfactory: even with only two points to be labeled, the $L_2$ solution would assign zero radius to one point (\autoref{fig:compare-opt-criteria}, center). The problem that we study in this paper is the $L_1$ version of this problem, in which we are given as input $n$ points and must maximize the sum of radii of disjoint disks centered at those points (\autoref{fig:compare-opt-criteria}, right). Two points can be optimally labeled by two equal-radius disks, and this criterion has the largest value of $p$ among $L_p$ criteria (maximizing sums of $p$th powers of the radii) that allow this equal-radius solution.

\begin{figure}[t]
\centering
\includegraphics[width=0.3\textwidth]{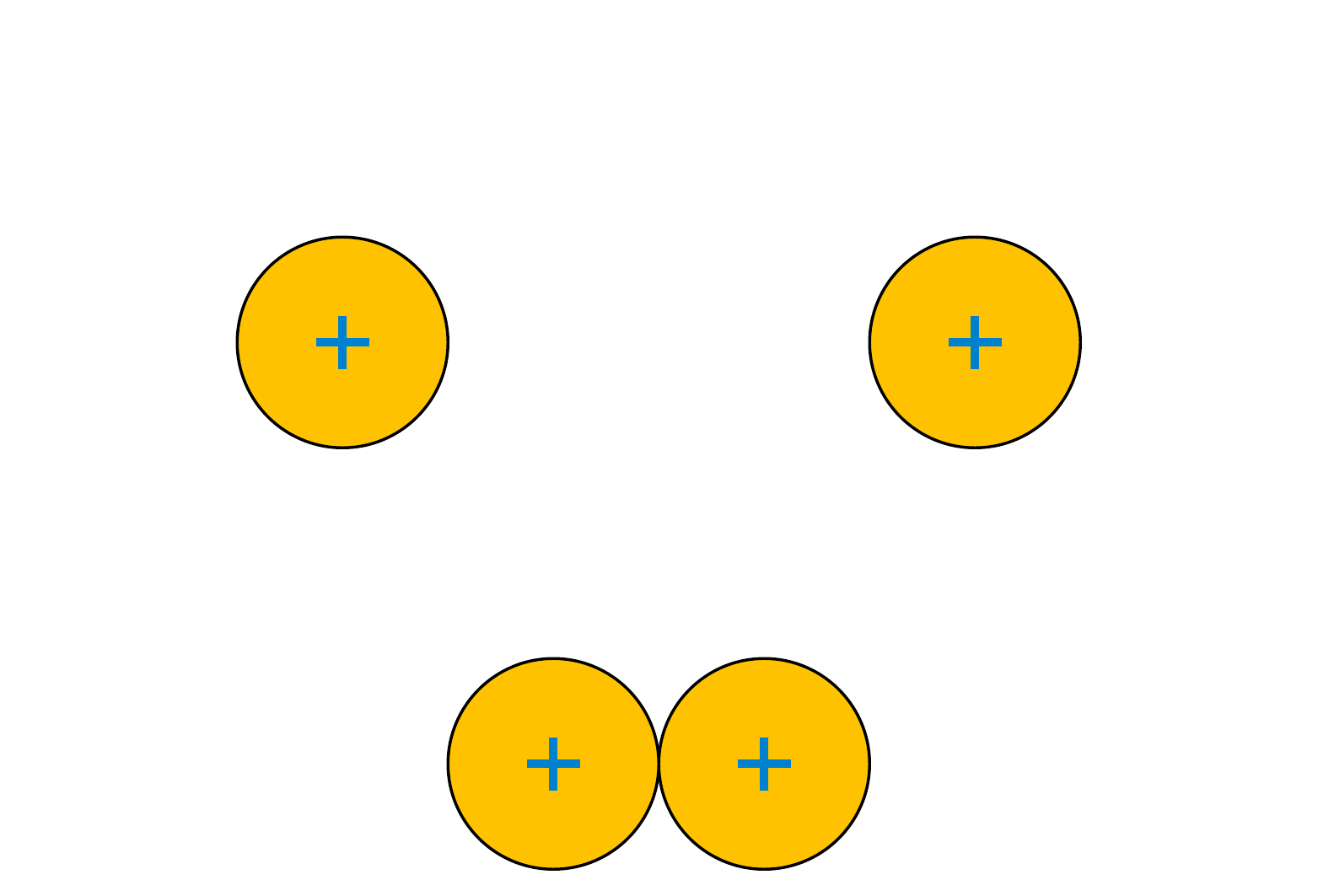}\quad
\includegraphics[width=0.3\textwidth]{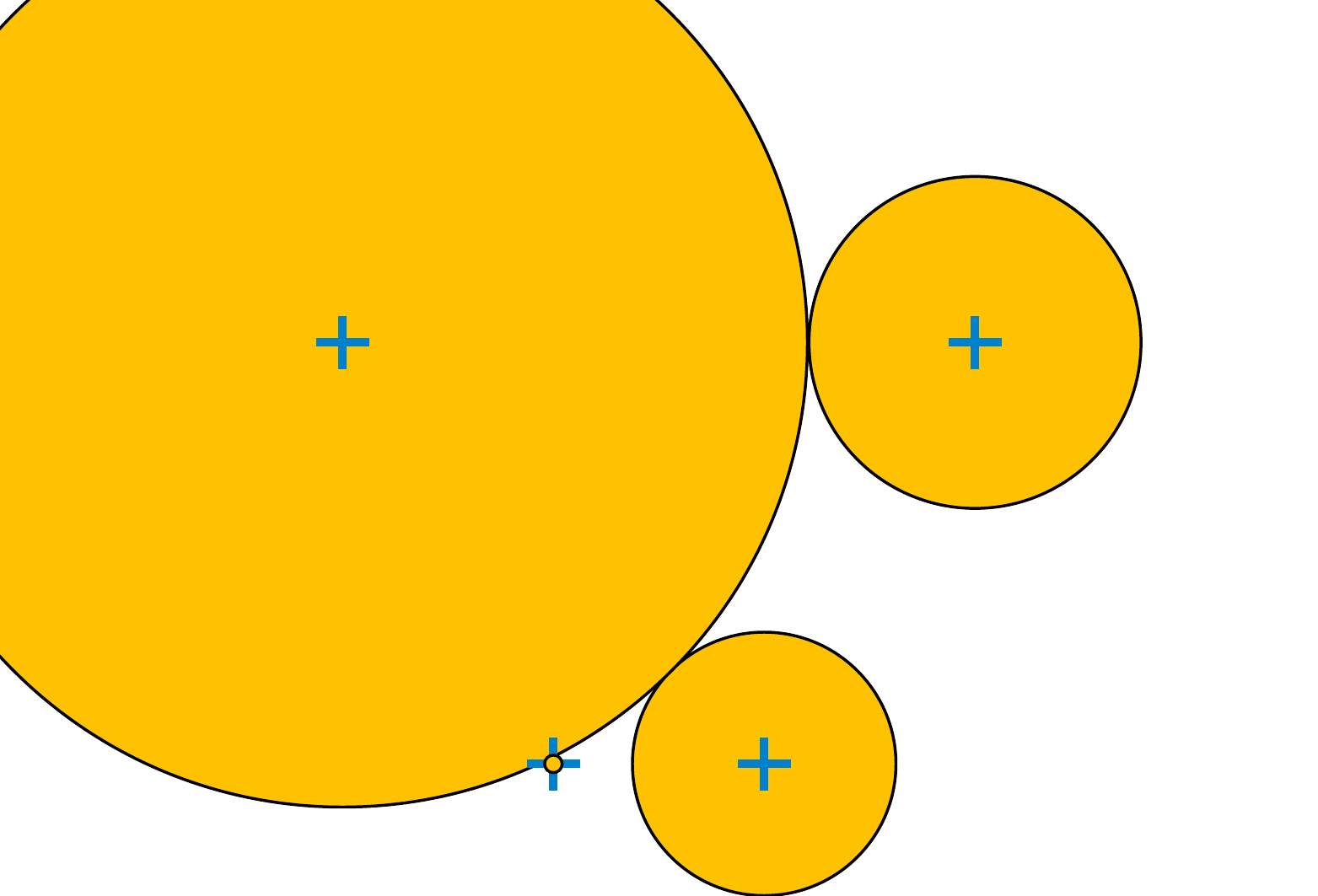}\quad
\includegraphics[width=0.3\textwidth]{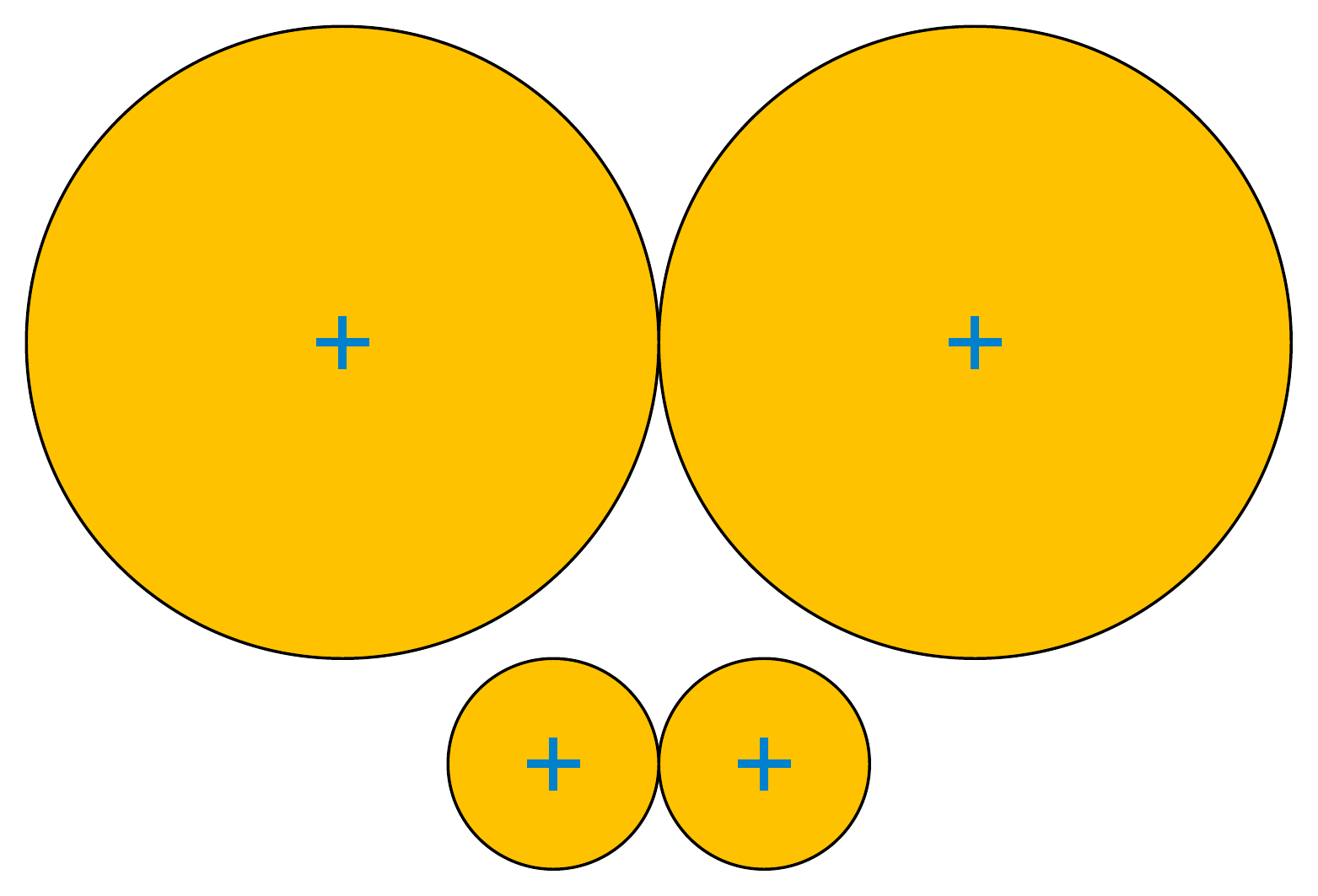}
\caption{For disks with fixed centers, maximizing the minimum radius (left) leads to disks that are smaller than necessary, while maximizing the sum of areas (center) leads to degenerate zero-radius disks. Maximizing the sum of radii (right) avoids both problems.}
\label{fig:compare-opt-criteria}
\end{figure}

We remark, however, that the algorithms we present here do not provide a complete solution to the map labeling problem. One reason for this is that there may be multiple optimal families of disks for a given collection of points, all with the same sum of radii, and that map labeling would require us to select a single radius for each point. For instance, with only two points, any pair of radii whose sum os the distance between the points would be optimal; in this case, the preferred solution for map labeling would likely be the one in which the two radii are equal. We leave the problem of how to formulate the best choice among the optimum solutions to our radius sum maximization problem, and of finding this best choice by a fast combinatorial algorithm, as open for future research.

Another, more serious, problem with applying radius sum optimization to map labeling is that, for more than two disks, the optimal solution may still involve disks of zero radius. In particular, when the input consists of three collinear points, the optimal solution assigns zero radius to the middle point, with the other two points having radii equal to their distance from the middle point (\autoref{fig:collinear}). One way to work around this issue, and ensure that all disks have non-zero radius, is to constrain all of the disks to have radius at least $\delta$, for some threshold value $\delta>0$. We may choose $\delta$ to be as large as half of the minimum distance between the given points. We then seek a system of radii $r_i\ge\delta$ that form non-overlapping disks of maximum total radius. As we show in \autoref{sec:constrained}, this constrained version of the problem, in which we have an additional lower bound on the radii of the balls, can be solved efficiently by transforming it into an unconstrained problem on a different distance metric.

\begin{figure}[t]
\centering\includegraphics[scale=0.5]{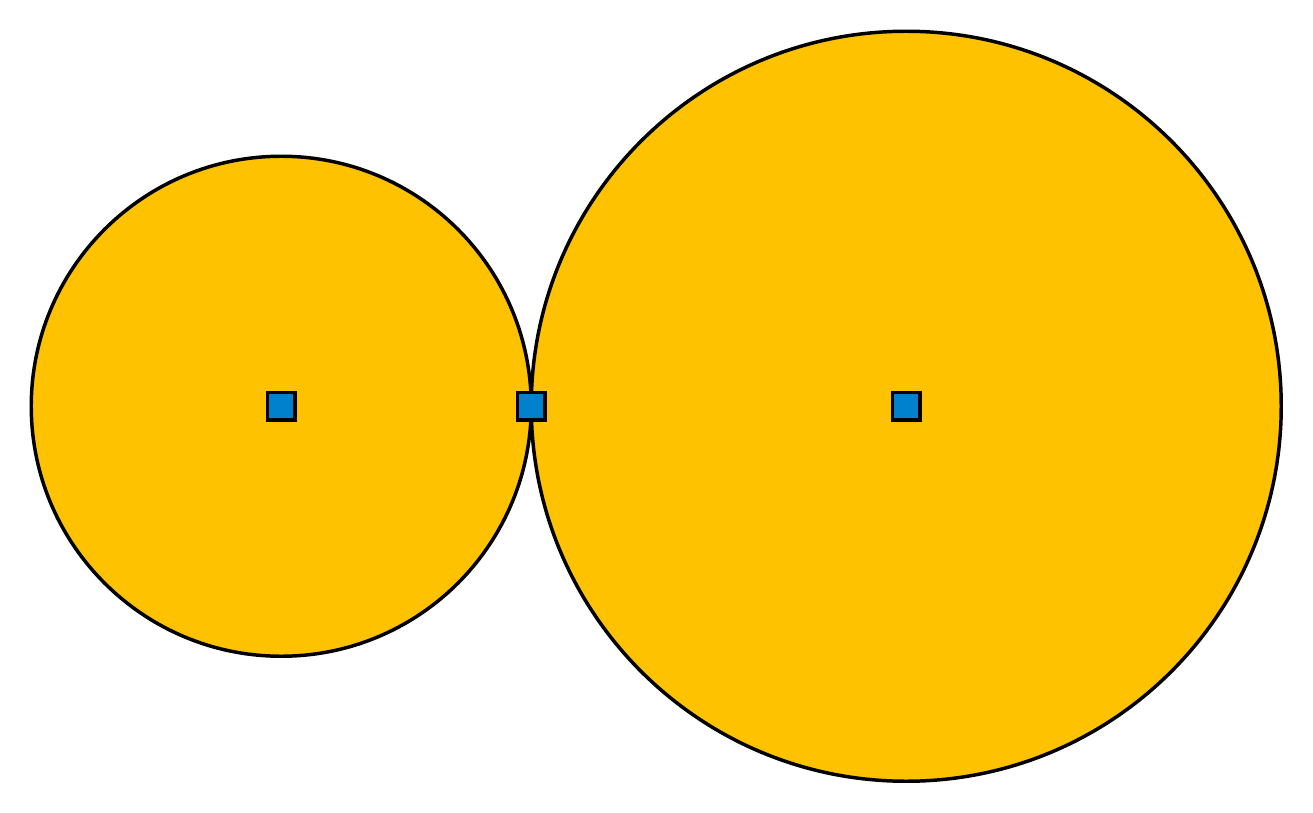}
\caption{The maximum radius sum for three collinear points assigns zero radius to the middle point.}
\label{fig:collinear}
\end{figure}

\subsection{Metric embedding into stars}

There has been extensive research on embedding complex metrics into simpler metrics
with low distortion; such methods have many applications in approximation algorithms,
by allowing approximations designed for the simpler metric to be applied to the more complex one~\cite{IndMat-HDCG-04}.
Eppstein and Wortman~\cite{EppWor-WADS-09} consider one such problem:
given any finite metric space $(X,d)$, find an embedding of it into a star metric space
with minimum distortion.
Here, a star is a space in which there is a central hub point $h$, not necessarily one of the given points, such that the distance between every two points is the sum of their distances to the hub.
The space of all distance-minimal non-contractive mappings from a given metric space $(X,d)$ to star metrics is called the \emph{tight span} of $(X,d)$. With the sup-norm (the maximum difference between hub distances) it is itself a metric space,
and includes an isometric copy of the original space $(X,d)$ (where each $x$ in $X$ is mapped to the star having $x$ as a hub)~\cite{Isb-CMH-64}.
The problem of Eppstein and Wortman is to select the optimal point of the tight span,
measuring the quality of each point by the distortion of the corresponding star metric.
We note that this is not the same as choosing an optimal hub point from some ambient space containing the input; for instance, if the input is the four points of a unit square in the Euclidean plane, the optimal star metric gives each of the four input points distance $1/2$ from the hub, while the best that can be achieved for a hub that itself belongs to the Euclidean plane would be distance $\sqrt{1/2}$.

Now consider a different and simpler optimization criterion: instead of minimizing distortion, suppose that we want to find a non-contractive mapping from a given metric space $(X,d)$ to a star network that minimizes the total or average distance of the points to the hub. That is, for each point $x_i$ of the input metric space, we should choose its distance $h_i$ to the hub, such that the mapping from $(X,d)$ to the resulting star metric space is non-contracting (for every $i$ and $j$, $d(i,j)\le h_i+h_j$)
and such that we minimize the sum of the chosen numbers $h_i$. Again, this can be viewed as choosing an optimal point from the tight span of $(X,d)$, but with a different optimization criterion.

\begin{figure}[t]
\centering\includegraphics[scale=0.5]{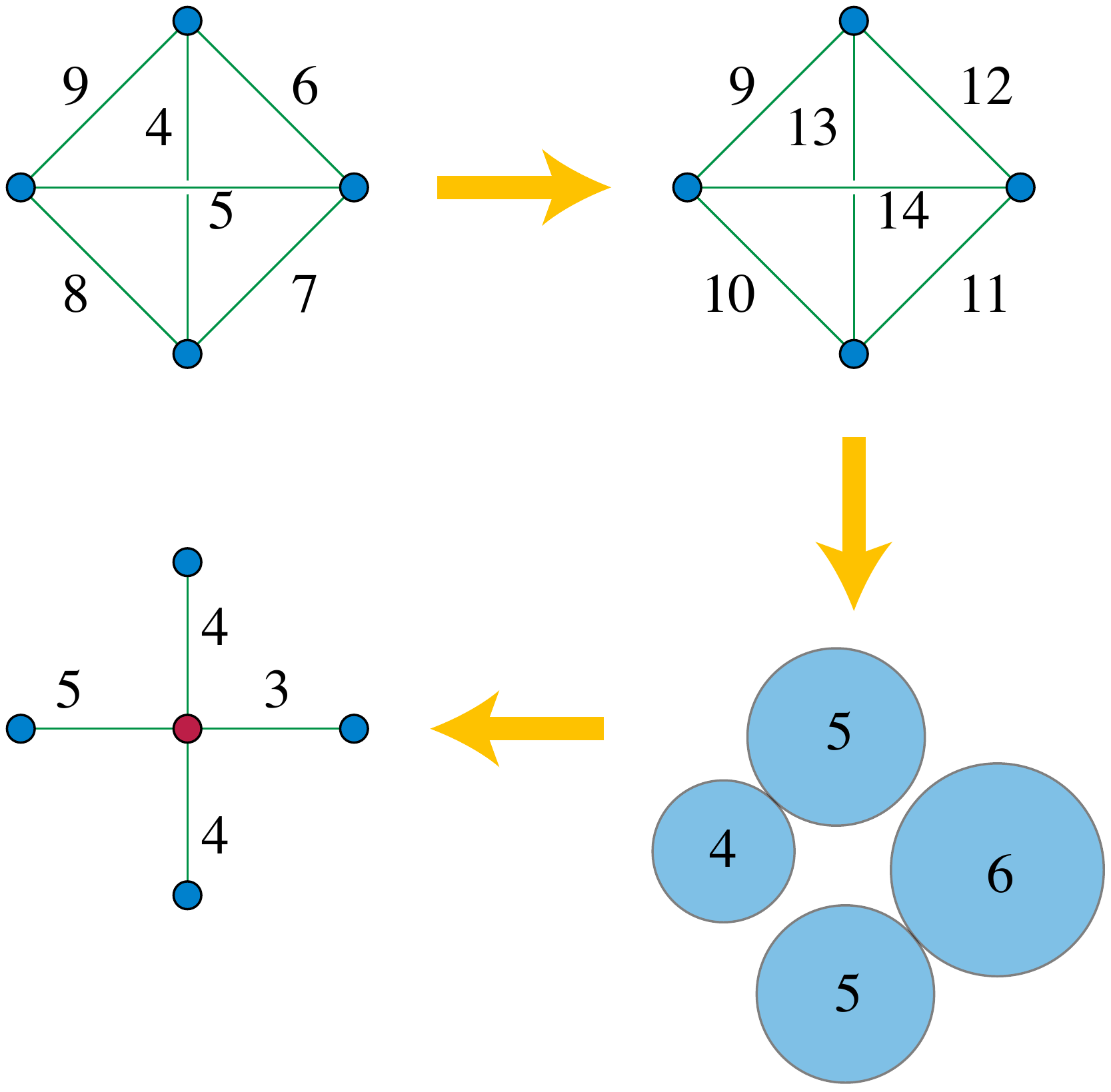}
\caption{Transforming a metric space $(X,d)$ (upper left) into the modified metric $(X,d^*)$ (upper right), finding non-overlapping balls with maximum sum of radii (lower right), and transforming the radii back into hub distances (lower left) produces an embedding of $(X,d)$ into a star metric minimizing the average distance from the hub.}
\label{fig:star-xform}
\end{figure}

As we now observe, this problem of choosing the star with the minimum average hub distance can be represented as a problem of minimizing the sum of radii of non-overlapping disks, for a different metric space $(X,d^*)$ on the same set of input points. To see this, let $D$ be the diameter of space $(X,d)$, the distance between its two farthest points. Note that, in the optimal star metric space, all hub distances will be at most $D$, for otherwise the solution would not be minimal.
We use $D$ to define the new distance $d^*(x,y)=2D-d(x,y)$. All new distances are between $D$ and $2D$, so this obeys the triangle inequality and the other requirements of a metric. Similarly, we convert hub distances $h_i$ to disk radii $r_i$ by the formula $r_i=D-h_i$ or equivalently $h_i=D-r_i$. Then two hub distances $h_i$ and $h_j$ obey the non-contractive inequality $d(i,j)\le h_i+h_j$ if and only if the corresponding two radii $r_i$ and $r_j$ obey the non-overlapping ball inequality $r_i+r_j\le d^*(x_i,x_j)$. Therefore, if we find a system of non-overlapping balls in $(X,d^*)$ minimizing the sum of radii of the balls, this will automatically also give us a system of hub distances $h_i$ defining a non-contractive mapping of $(X,d)$ to a star network that minimizes the average distance of a point to the hub. These transformations are illustrated in \autoref{fig:star-xform}.

If we use these transformations together with the algorithms that we describe in this paper, we can find the optimal star network embedding for any $n$-point metric space, minimizing the average distance to the hub, in time $O(n^3)$. This compares favorably with the $O(n^3\log^2 n)$ time of Eppstein and Wortman~\cite{EppWor-WADS-09} for finding the star network embedding that minimizes the distortion. 

\section{Equivalence to cycle cover}

Although our radius-sum optimization problem is not equivalent to geometric matching, we show in this section that it is equivalent by linear programming duality to a related problem of finding a cycle cover in an associated geometric graph. We will use this equivalence in formulating algorithms for both problems.

Let $p_i$ ($i=0,\dots n-1$) be a set of points in a metric space, with distances $d(p_i,p_j)$.
The problem of finding non-overlapping balls $(p_i,r_i)$ that maximize the sum of the radii can be expressed as a linear program:
\[
\text{maximize~} \sum r_i
\]
subject to the inequality constraints
\begin{align*}
\forall i:\quad& r_i \ge 0\\
\forall i,j:\quad& r_i+r_j \le d(p_i,p_j).
\end{align*}
By linear programming duality \cite[Ch.~12]{Vaz-AA-01} this has the same value as the dual linear program:
\[ \text{minimize~} \sum w_{ij} d(p_i,p_j) \]
subject to the inequality constraints
\begin{align*}
\forall i,j:\quad& w_{ij} \ge 0 \\
\forall i:\quad& \sum_j w_{ij} \ge 1
\end{align*}
That is, we must find non-negative weights $w_{ij}$ for the edges of a complete geometric graph such that each vertex has incident edge weights totaling at least one, and minimizing the weighted sum of edge lengths.

\begin{lemma}
\label{lem:dual-equality}
For distances obeying the triangle inequality,
there is an optimal solution to the dual linear program described above in which the incident edge weights at each vertex total exactly one.
\end{lemma}

\begin{proof}
Define the excess of a vertex to be the total weight of its incident edges minus one.
If any edge has weight greater than one, its weight can be reduced to exactly one, improving the quality of the solution without changing its feasibility. Otherwise, if any vertex has positive excess, we can subtract an equal positive weight from two of its incident edges and add the same weight to the edge between the other endpoints of these edges. By the triangle inequality, this change does not worsen the solution, and again it does not change its feasibility. By making such changes we can reduce the total excess, maintaining feasibility, until eventually all vertices have excess zero.
\end{proof}

This dual program (either in the form first given above, or with the restriction that the weights at each vertex sum to exactly one according to \autoref{lem:dual-equality}) is the linear programming relaxation of the problem of finding a minimum weight perfect matching in the complete geometric graph. Such a relaxation is primarily used for bipartite graphs, for which it is exact~\cite[Ex.~12.7]{Vaz-AA-01}. For non-bipartite graphs such as the complete graph, it has a half-integral optimal solution in which all weights $w_{ij}$ belong to $\{0,1/2,1\}$~\cite[Ex.~14.8]{Vaz-AA-01}. Doubling the weights to make them integers, and interpreting each doubled weight as an edge multiplicity, gives us a combinatorial description of the dual solution as a multiset of edges in which (by \autoref{lem:dual-equality}) each vertex has degree two. That is, we have the following result:

\begin{definition}
In a weighted graph $G$, a \emph{cycle cover} is a multigraph that uses only edges of $G$ (possibly using some edges twice), and has exactly two edges incident to every vertex of $G$. The \emph{weight} of the cycle cover is the sum of weights of its edges (counted with multiplicity).
A \emph{minimum cycle cover} is a cycle cover that has the minimum possible weight among all cycle covers of~$G$.
\end{definition}

\begin{theorem}
\label{thm:dual}
The maximum sum of radii of nonoverlapping balls, centered at points $p_i$ of a metric space, equals half of the weight of a minimum cycle cover of the complete geometric graph on the points $p_i$ with edge lengths equal to the distances between the edge endpoints.
\end{theorem}

\begin{proof}
This follows immediately from the facts that, as with any linear program, the linear program expressing the problem of maximizing the sum of radii (the maximization problem at the start of this section) has the same value as its dual minimization problem, from the known fact that the specific linear program given by this dual minimization problem has a half-integer solution, and from \autoref{lem:dual-equality} which allows us to interpret this half-integer solution as a choice of two incident edges at each vertex.
\end{proof}

For example, the example of three collinear points, given in \autoref{sec:label} and \autoref{fig:collinear} as an input that can cause one ball to have zero radius, has a single three-point cycle as its minimum cycle cover. The length of this cycle, the sum of the three distances between the points, equals twice the sum of the radii of the optimal balls (zero for the middle point, and the distance to the middle point for the two outer points).

\begin{figure}[t]
\centering\includegraphics[scale=0.5]{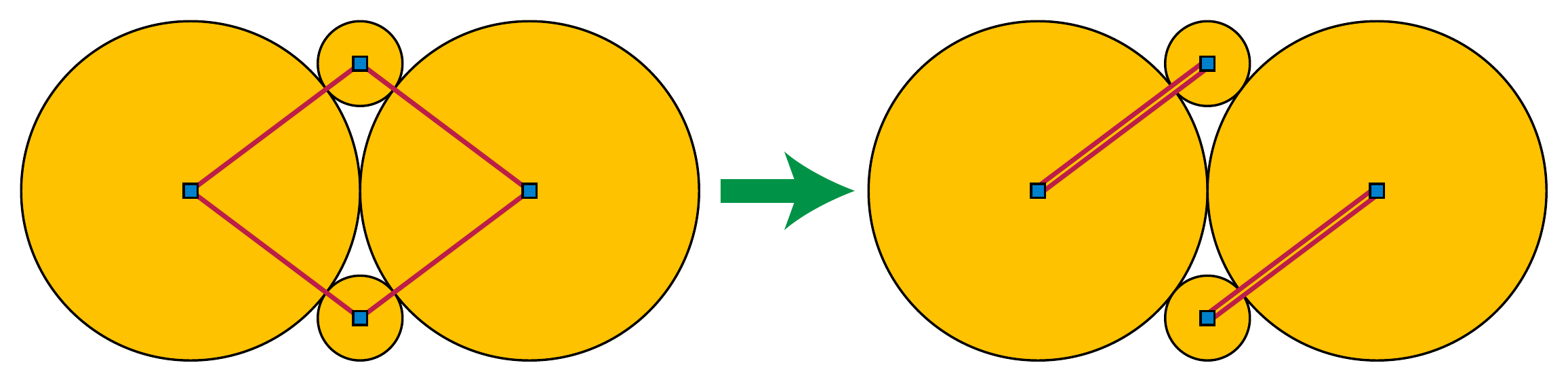}
\caption{Eliminating a long even cycle (here, a rhombus) in the minimum cycle cover by splitting it into 2-cycles with no greater length.}
\label{fig:split-even}
\end{figure}

There always exists a minimum cycle cover in which all cycles of more than two edges have odd length. To see this, observe that any long even cycle $C$ can be partitioned into two disjoint matchings, and that at least one of these matchings gives a covering of the same vertices as $C$ by $2$-cycles whose total weight is at most that of $C$ (\autoref{fig:split-even}). Therefore, from now on we will assume that our cycle covers have no long even cycles.

\autoref{thm:dual} also gives us an easy-to-test optimality condition for the maximum sum of radii problem, that applies to inputs in general position (without extra touching balls beyond the ones required by the solution):

\begin{definition}
The \emph{touching graph} of a family of nonoverlapping balls has a vertex for each ball and an edge connecting each two balls that touch.
\end{definition}

\begin{corollary}
\label{cor:opt}
Suppose that no two balls in a given family of balls overlap, and that each connected component of the touching graph of the balls is an odd cycle or an isolated edge. Then this family of balls has the maximum sum of radii of any family with the same centers.
\end{corollary}

\begin{proof}
The touching graph of the balls (viewing each isolated edge as a $2$-cycle) gives a cycle cover of length half the sum of radii of the given balls. By \autoref{thm:dual} no cycle cover can be shorter and no system of balls with the same centers can have a larger sum of radii.
\end{proof}

In linear programming terms, this corollary is an instance of \emph{complementary slackness}, a more general condition for testing whether a simultaneous primal and dual solution are both optimal.
\autoref{fig:OptimalDisks} shows a family of disks meeting the conditions of the corollary, together with their touching graph.

\begin{figure}[b]
\centering\includegraphics[width=0.5\textwidth]{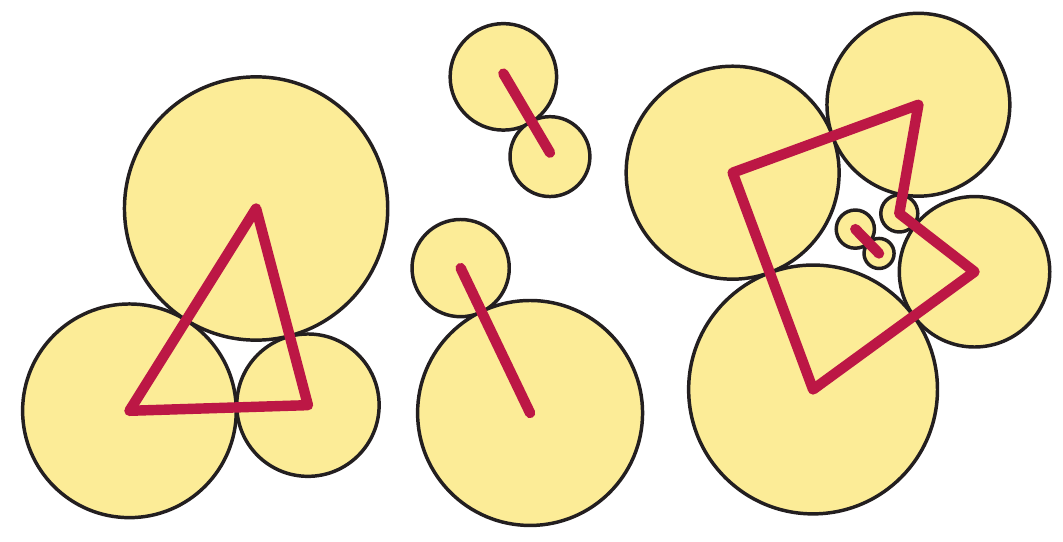}
\caption{Nonoverlapping disks that (by \autoref{cor:opt}) maximize the sum of radii for their centers.}
\label{fig:OptimalDisks}
\end{figure}

\section{Cycle covers from matchings}

\autoref{thm:dual}, giving an equivalence between maximizing the sum of radii and minimizing the total length of a cycle cover, does not yet give us a combinatorial algorithm for either of these two equivalent problems. And because it only gives us an equivalence of the optimum values of the two problems, it does not help us compute the actual radii in an optimum collection of disks.
Therefore, we still need to show how to solve both problems combinatorially, in a way that provides for us the solution itself and not just the solution value.

We begin with the cycle cover. As we show below, it can be transformed into a more familiar matching problem using the notion of a bipartite double cover.

\begin{definition}
The \emph{bipartite double cover} $2G$ of a graph $G$ is the tensor product $G\times K_2$, a bipartite graph with two copies of each vertex of $G$ (one of each color) and two copies of each edge of $G$ (one for each pair of oppositely-colored copies of the endpoint of the edge). If $G$ is weighted, we use the same edge weights in $G$ and $2G$.
\end{definition}

\begin{lemma}
\label{lem:cover2match}
Every perfect matching in $2G$ corresponds (under the mapping that takes each vertex of $2G$ to the corresponding vertex in $G$) to a cycle vertex cover with equal total length in $G$. Every cycle vertex cover in $G$ comes from a perfect matching in $2G$ in this way.
\end{lemma}

\begin{proof}
To convert a matching in $2G$ to a cycle cover in $G$, we replace every edge of the matching by the corresponding edge in $G$. Each vertex in $G$ belongs to two such edges (one for the matched edge of each of its two copies in $2G$) so the resulting multiset of edges has degree two at every vertex of $G$. I.e., it is a cycle cover.

To convert a cycle cover in $G$ to a matching in $2G$ we orient each cycle of the cycle cover consistently (choosing one of the two possible orientations of each cycle, arbitrarily). Then,
in $2G$, we match each vertex in one copy of $G$ to its outgoing neighbor in the other copy of $G$.

These two transformations clearly preserve the solution value, proving the lemma.
\end{proof}

See \autoref{fig:DoubledGraph} for an example.
Each cycle of more than two vertices in $G$ has two different representations as a set of matched edges in $2G$, but this ambiguity is not a problem.
We can find a minimum cycle vertex cover in~$G$ by finding a minimum weight perfect matching in~$2G$.

A minimum weight perfect matching on a bipartite graph of $n$ vertices and $m$ edges can be found in time $O(mn+n^2\log n)$~\cite{FreTar-JACM-87}. (Better times are known when the weights are small integers~\cite{DuaSu-SODA-12}.) By solving this problem on the doubled complete geometric graph, the optimal sum of radii can be computed in time $O(n^3)$. However, this is still slower than we would like for Euclidean spaces, and does not yet tell us the radii of the individual balls.

\begin{figure}[b]
\centering\includegraphics[width=0.5\textwidth]{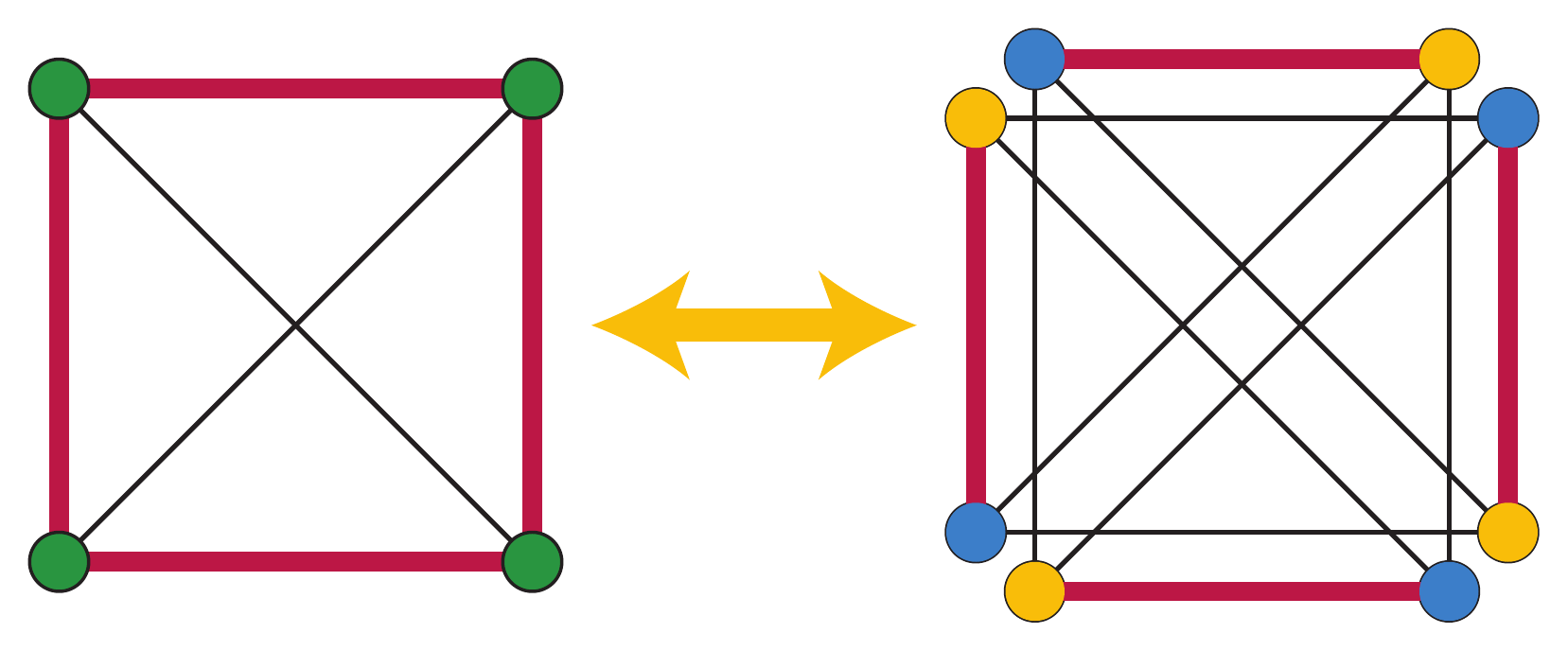}
\caption{The correspondence between a cycle cover of a graph $G$ (left, in this case, a complete graph $K_4$) and a matching in $2G$ (right).}
\label{fig:DoubledGraph}
\end{figure}

\section{From cycle covers to balls}

Each odd cycle in a minimum cycle cover of the complete geometric graph corresponds to a unique system of balls that maximizes the sum of radii for those points. Let the cycle have vertices $p_0, p_1, \dots p_{k-1}$ and edge lengths $\ell_i = d(p_i,p_{i+1\bmod k})$. Then for any $j$ with $0\le j<i$  set the radius $r_j$ of point $p_j$ to be
\[ r_j=\sum_i \frac{\pm \ell_i}{2}, \]
where we choose the signs in this sum so that the two edges adjacent to $p_j$ have positive sign and every other point $p_i$ is incident to two edges with opposite signs from each other, as depicted in \autoref{fig:radius-assgt}.

\begin{figure}[hbt]
\centering\includegraphics[scale=0.4]{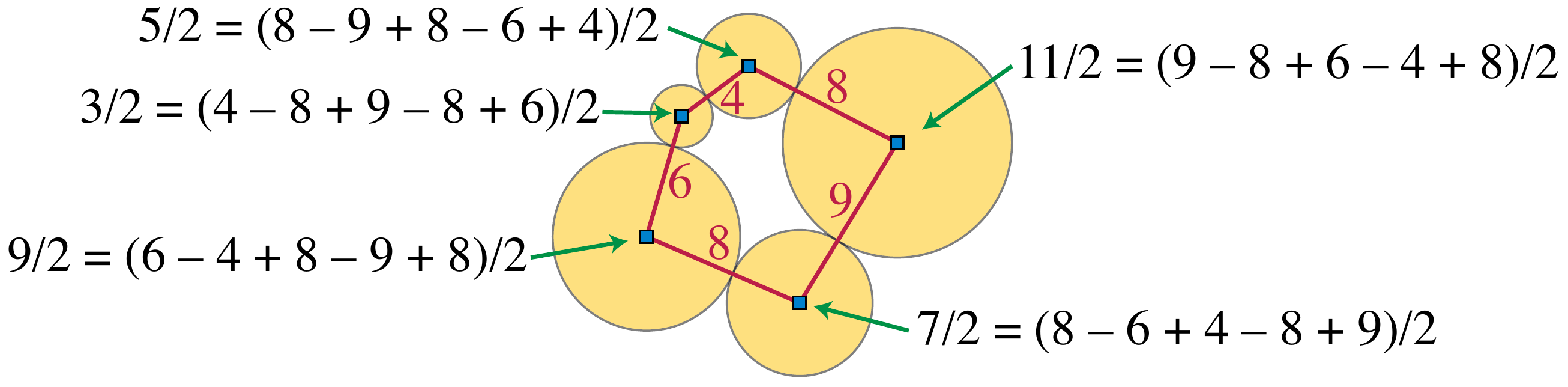}
\caption{Converting the edge lengths of an odd cycle to radii of non-overlapping disks.}
\label{fig:radius-assgt}
\end{figure}

\begin{lemma}
\label{lem:odd-rad}
The radius assignment given above gives the unique set of balls with maximum sum of radii centered at the vertices of an odd cycle in a minimum cycle cover.
\end{lemma}

\begin{proof}
For every edge $e$ of the cycle, the two radii of the disks centered at the endpoints of $e$ have expressions in which the terms for $e$ both have positive sign. However, these same two expressions differ in their choice of the signs for every other edge.
Therefore, when we sum these two radii, every term except the ones for the lengths of $e$ cancels, and the sum of the two radii is exactly $e$ Thus, the two balls centered at the endpoints of $e$ touch but do not overlap, and the sum of all the radii is half the length of the cycle, as desired.
We have not yet shown that balls for non-consecutive cycle vertices are non-overlapping, but that will follow from uniqueness: there is no other optimal solution, so this must be a valid solution.

If any other system of non-overlapping balls wth the same centers includes a consecutive pair that do not touch, then its sum of radii cannot be as large. For, in any system of nonoverlapping balls, the length of each cycle edge equals the sum of radii of nonoverlapping balls at its endpoints, plus a non-negative ``gap'' equal to the difference between the length and the sum of radii. So if any gap is positive, the doubled sum of radii will fall short of the cycle length by that gap.

No other system of non-overlapping balls with the same centers can have each consecutive pair touching. For, increasing the radius of any one ball (relative to the solution given above) would cause a chain of alternating decreases and increases of the radii of all the other balls, of equal magnitude, leading to an inconsistency because the cycle length is odd and the changes in radius cannot strictly alternate in sign.
\end{proof}

The $2$-cycles in the cycle cover cause us more trouble, because their radii may be constrained by other nearby points but are not in general uniquely determined. For instance, in \autoref{fig:TwoCycleConstraints}, each of the three pairs of touching circles must have nearly-equal radii to avoid overlaps with nearby circles.

\begin{figure}[t]
\centering
\includegraphics[width=0.5\textwidth]{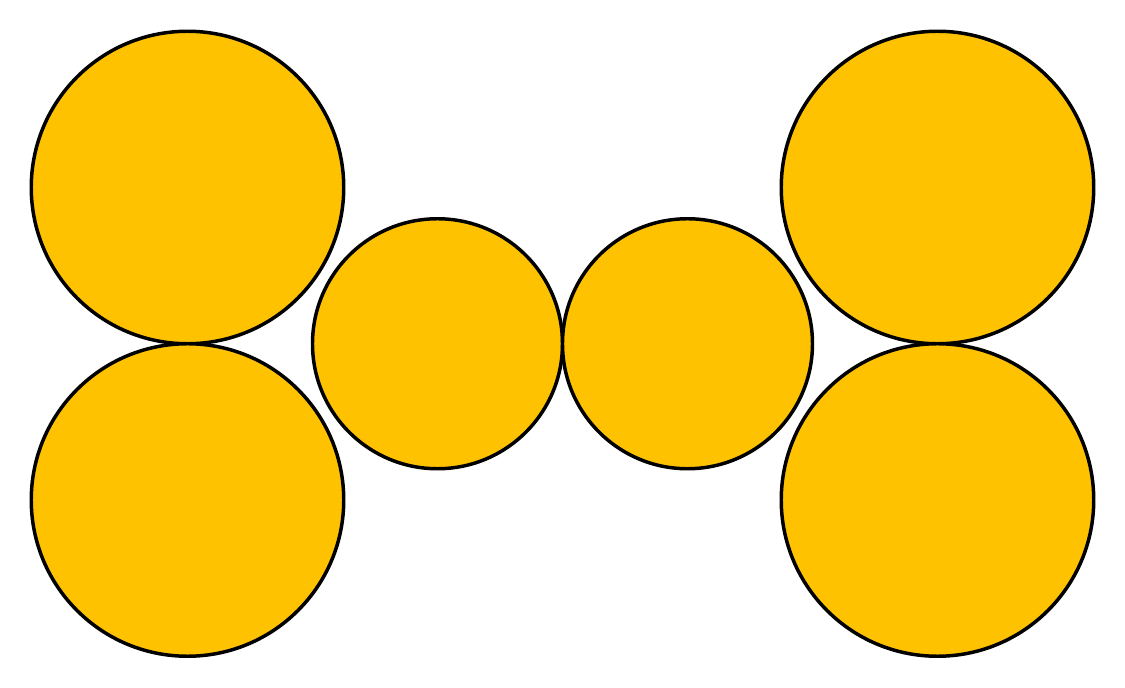}
\caption{Pairs of disks corresponding to $2$-cycles in the minimum cycle cover may constrain the radii of other nearby pairs of disks. In this example, each of the three pairs of touching circles must have nearly-equal radii to avoid overlaps with nearby circles.}
\label{fig:TwoCycleConstraints}
\end{figure}

To describe how we transform a minimum weight perfect matching problem on $2K_n$ into a feasible system of balls with maximum sum of radii, we need to look more deeply into the details of the Hungarian algorithm for minimum weight perfect matching.
The version of this algorithm described by Tarjan~\cite[Secs.~8.4 \&~9.1]{Tar-DSNA-83} maintains a matching (initially empty) which it uses to orient the graph. Unmatched edges are directed from one side of the bipartition to the other, and matched edges are directed in the opposite direction. In each iteration the algorithm finds a minimum-cost alternating path between two unmatched vertices; here, the cost of a path is the sum of lengths of unmatched edges, minus the sum of lengths of matched edges. It uses this path to extend the matching by one edge.

Because matched edges are subtracted from the path length, each path search involves negative-weight edges. However, the algorithm also maintains a system of non-negative weights that are equivalent (in the sense of having the same shortest paths), allowing Dijkstra's algorithm to be used, and adjusts these weights to keep them non-negative after each iteration. Each iteration increases the number of matched edges, so there are $O(n)$ iterations. The time per iteration can be bounded by the time for Dijkstra's algorithm. Using Fibonacci heaps this gives a total running time of $O(mn + n^2\log n)$~\cite{FreTar-JACM-87}.

To adjust edge weights  we maintain dual variables: a real number for each vertex of the bipartite graph. We subtract these variables from the length of each incident unmatched edge, and add these variables to the (negated) length of each incident matched edge. In order for the adjusted edge weights to be non-negative, the algorithm maintains an invariant that the length of each unmatched edge is at least the sum of the dual variables at its endpoints, and each matched edge length equals the sum of the dual variables at its endpoints.

Recall that the graph $2K_n$ to which we apply this algorithm has two vertices for each input point $p_i$, one of each color (say red and blue). We may visualize the dual variables at these two vertices as two balls (again, one red and one blue) both centered at point $p_i$. Then the invariants maintained by the Hungarian algorithm can be expressed in our terms as stating that pairs of balls of opposite colors cannot overlap unless they have the same center, and that each matched edge comes from a touching pair of oppositely colored balls (\autoref{fig:HungarianDuals}). However, this visualization may be somewhat misleading, as we allow the dual variables to be negative.

\begin{figure}[ht]
\centering
\includegraphics[width=0.5\textwidth]{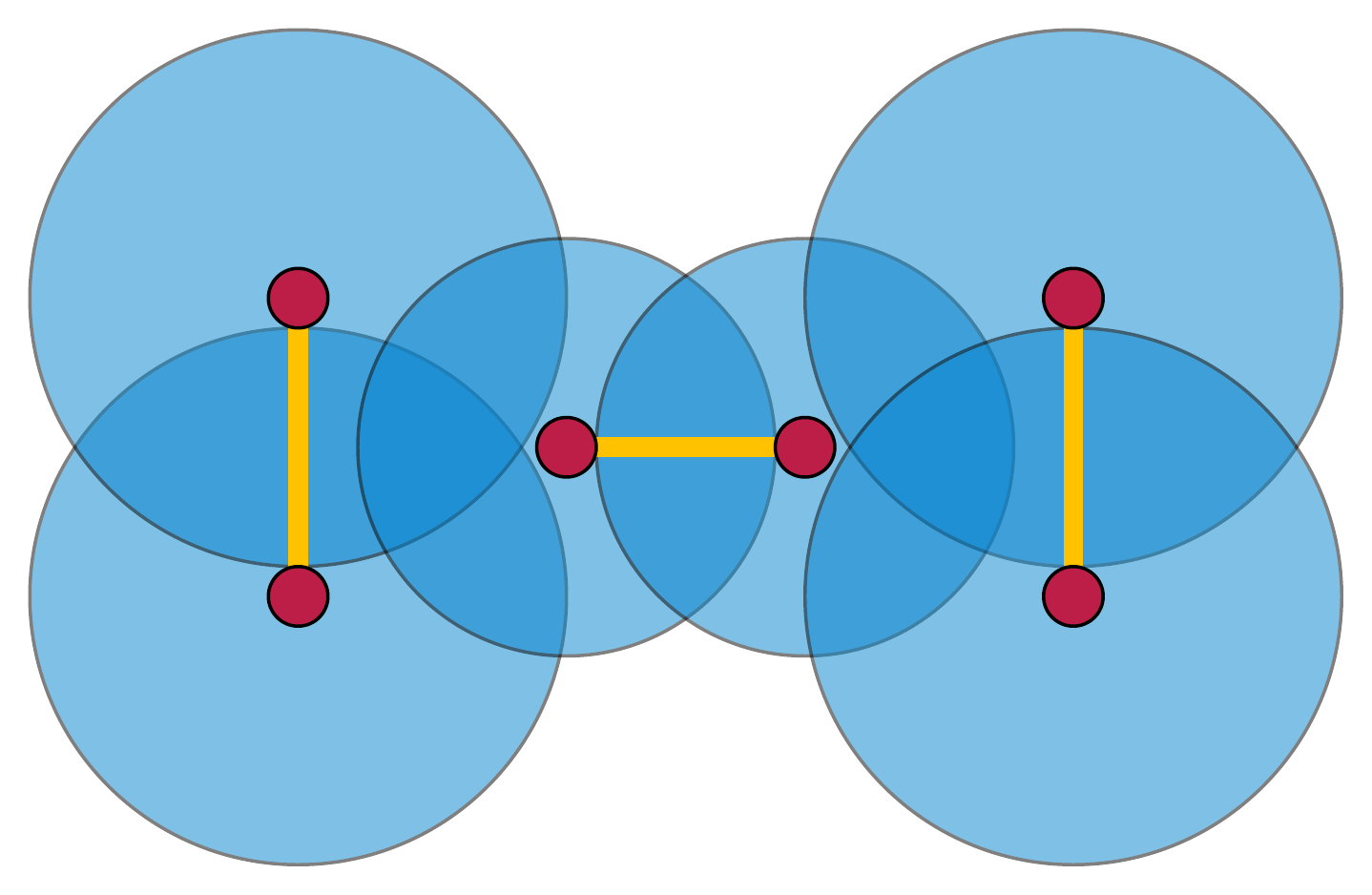}
\caption{The dual variables for the Hungarian matching algorithm on the bipartite graph $2K_n$ can be visualized as red and blue balls with no bichromatic overlaps, touching at each matched pair.}
\label{fig:HungarianDuals}
\end{figure}

\begin{lemma}
\label{lem:avg-radius}
If we are given the dual variables found by the Hungarian algorithm for minimum weight perfect matching in the bipartite graph $2K_n$, we can construct from them a system of nonoverlapping balls with maximum sum of radii, in linear time.
\end{lemma}

\begin{proof}
Let the two dual variables for point $p_i$ be $a_i$ and $b_i$, and calculate the radius of the ball centered at $p_i$ to be the average $r_i=(a_i+b_i)/2$ of the two dual variables.
For each $i$ and $j$, we have (by assumption) $a_i+b_j\le d(p_i,p_j)$ and $b_i+a_j\le d(p_i,p_j)$, and averaging these two inequalities gives us that $r_i+r_j\le d(p_i,p_j)$. That is, the balls with radius $r_i$ are non-overlapping.

Each edge of the cycle cover of $K_n$ has two balls of opposite color at its endpoints whose radii sum to the edge length. For each cycle of the cycle cover of $K_n$, each ball centered at a cycle vertex will contribute to this equality for exactly one of the two incident cycle edges, so the sum of the radii of the blue and red balls centered at the cycle vertices equals the length of the cycle. Therefore, when we average these red and blue radii to give the radii $r_i$ of a single system of balls, the sum of the averaged radii equals half the length of the cycle. Thus, over the whole input,
the sum of all the radii $r_i$ equals half the length of the cycle cover. The optimality of this system of nonoverlapping balls then follows from \autoref{thm:dual}.
\end{proof}

Combining \autoref{lem:cover2match}, \autoref{lem:avg-radius}, and the known algorithms for minimum weight perfect matching in bipartite graphs gives us the following result.

\begin{theorem}
Given $n$ points in an arbitrary metric space, we can construct a system of nonoverlapping balls centered at those points, maximizing the sum of radii of the balls, in time $O(n^3)$.
\end{theorem}

\begin{proof}
We construct the bipartite graph $2K_n$, weighted by the distances between points, and find a minimum-weight perfect matching in $2K_n$ (or equivalently by \autoref{lem:cover2match} a cycle vertex cover in $K_n$), by using an algorithm for minimum weight perfect matching that also provides the values of the dual variables (as the Hungarian algorithm does).
We then apply \autoref{lem:avg-radius} to convert these dual variable values to an optimal system of radii.
\end{proof}

\section{Euclidean speedup}

To speed up the search for balls with maximum sum of radii in low-dimensional Euclidean spaces, we replace the complete geometric graph of the previous sections with a sparser graph that contains the optimal cover and behaves the same with respect to testing whether systems of balls are nonoverlapping.

\subsection{The nearest-neighbor overlap graph}

\begin{definition}
Given a system of points $p_i$ in a metric space, let the \emph{nearest neighbor distance} $\delta_i$ of each point $p_i$ be
\[ \delta_i = \min_{j\ne i} d(p_i,p_j). \]
\end{definition}

In Euclidean spaces of bounded dimension, these distances can be found for all points in total time $O(n\log n)$~\cite{Cla-FOCS-83,Vai-DCG-89}.

\begin{definition}
Define the \emph{nearest-neighbor overlap graph} $\NNG$ (\autoref{fig:NNG}) to be the graph whose vertices are the points $p_i$, with an edge between two points $p_i$ and $p_j$ when the balls with radius equal to the nearest neighbor distance overlap or touch. That is, $p_i$ and $p_j$ are adjacent exactly when
\[ d(p_i,p_j) \le \delta_i + \delta_j. \]
\end{definition}

\begin{figure}[ht]
\centering\includegraphics[width=0.5\textwidth]{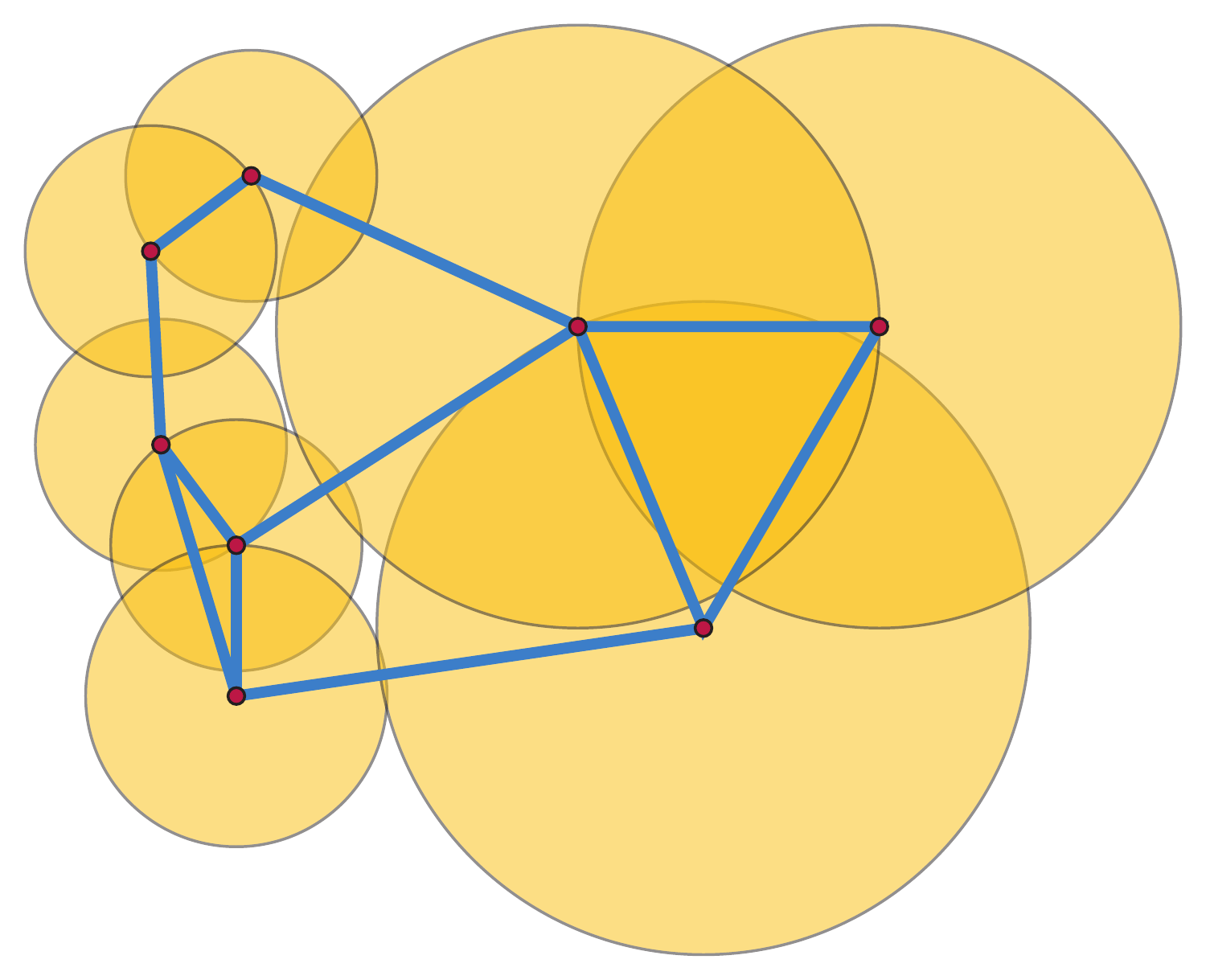}
\caption{The nearest neighbor balls of a set of points and their nearest neighbor overlap graph $\NNG$.}
\label{fig:NNG}
\end{figure}

\begin{lemma}
\label{lem:cover-subgraph}
Each edge of the minimum cycle cover of the complete geometric graph belongs to $\NNG$.
\end{lemma}

\begin{proof}
Each edge of the cycle cover can be represented by touching balls in a system of nonoverlapping balls.
Because each ball centered at $p_i$ in any system of nonoverlapping balls centered at the given points has radius at most $\delta_i$ (else it would overlap the ball of the nearest neighbor), two balls of such a system can touch only if the larger balls with radius $\delta_i$ touch or overlap.
\end{proof}

\begin{lemma}
\label{lem:NNG-conflicts}
A system of balls centered at the points $p_i$ is nonoverlapping if and only if, for each edge of $\NNG$, the two endpoints of the edge are nonoverlapping.
\end{lemma}

\begin{proof}
A ball centered at $p_i$ with radius greater than~$\delta_i$ overlaps the ball centered at the nearest neighbor of $p_i$, which is necessarily adjacent to $p_i$ in $\NNG$. Otherwise, if all balls have radius at most $\delta_i$, then they can overlap only if their centers are adjacent in $\NNG$.
\end{proof}

By \autoref{lem:cover-subgraph} we can find our minimum cycle cover by applying a weighted matching algorithm to $2\NNG$. The system of radii formed by averaging the dual variables of the matching in $2\NNG$ will give us an optimal set of nonoverlapping balls, just as it did in $2K_n$. For, the averaged radii will achieve the optimal value and will have no overlaps that are detected by the edges of $\NNG$; by \autoref{lem:NNG-conflicts} there can be no other overlaps. It remains to show that we can find a matching (and its dual variables) quickly in this graph. For this we will need some geometry, since in arbitrary metric spaces $\NNG$ can be complete.

\subsection{Weighted matching in separated bipartite graphs}
In this section we prove the existence of an efficient minimum weight perfect matching algorithm for bipartite graphs obeying a separator theorem. We will later show that the nearest-neighbor graph $\NNG$ and its bipartite double cover obey such a theorem in low-dimensional Euclidean spaces, allowing us to apply this algorithm.

\begin{definition}
Let $s$ be a real number in the range $0<s<1$.
Define an $s$-\emph{separator} of an $n$-vertex graph to be a subset $S$ of the vertices the removal of which allows the graph to be partitioned into two subgraphs, disconnected from each other and each having at most $sn$ vertices.
When $s$ is a constant (depending on some general family of graphs but not the specific graph in that family) we will omit it and call $S$ a \emph{separator}.
 \end{definition}

An efficient algorithm for weighted matching in graphs with small separators was already given by Lipton and Tarjan
\cite[Sec.~7]{LipTar-SICOMP-80}. They directly consider only planar graphs, but their method would also work for other graphs that have separator hierarchies.
However, they consider a different variant of matching, maximum weight matching on non-bipartite graphs, and are not explicit in how they maintain and update the analogue for that problem of the dual variables that we need. They also do not take advantage of more recent advances in shortest path algorithms for graphs with separators~\cite{HenKleRao-JCSS-97}. Here we describe an algorithm based on similar ideas to the Lipton--Tarjan algorithm that maintains the dual variables that we need and is faster by a logarithmic factor.

Suppose we are given a weighted graph $G$ with a separator $S$ (part of a separator hierarchy for $G$), such that $G\setminus S$ can be partitioned into two subgraphs $L$ and $R$, disconnected from each other, with $L$ and $R$ both having a number of vertices that is smaller by a constant factor.
For such a graph, Lipton and Tarjan~\cite[Sec.~7]{LipTar-SICOMP-80} describe a divide-and-conquer algorithm for maximum-weight matching (an equivalent problem to minimum-weight perfect matching)
that recursively matches $L$ and $R$, and then adds the separator vertices one at a time to the union of $L$ and $R$. When each separator vertex is added, the maximum matching may change along a single alternating path.

Lipton and Tarjan write that ``given a suitable representaton'', each alternating path can be found in the time for a single application of Dijkstra's algorithm. This, in turn, allows the most expensive term in the running time of their algorithm to be bounded by the product of 
the time for finding each path with the number of separator vertices. The ``suitable representation'' that they refer to is made more complicated by the possibility that $G$ might not be bipartite, but in the bipartite case it is essentially the system of dual variables that we need. However, their handling of these variables is lacking in detail, so we provide here a more explicit divide and conquer matching algorithm, with the following differences:

\begin{itemize}
\item We formulate the problem as finding a minimum-weight maximum-cardinality matching in a graph with non-negative edge weights. This slight generalization of minimum-weight perfect matching allows us to perform recursive calls without needing to guarantee that each subproblem has a perfect matching, and simplifies the definition of the dual variables relative to their definition for maximum-weight matching.
\item We recurse on the induced subgraphs $G[S\cup L]$ and $G[S\cup R]$ rather than on $L$ and $R$. This inclusion of the separator vertices in the recursive subproblems results in the calculation of two inconsistent systems of dual variables for these vertices, which must be reconciled when the subproblem solutions are combined.
\item We compute shortest paths using an algorithm of Henzinger et al.~\cite{HenKleRao-JCSS-97} rather than by Dijkstra's algorithm. This algorithm takes linear time per path, after a separator hierarchy has already been computed, saving a logarithmic time factor compared to Dijkstra.
\item We consider only bipartite graphs rather than the more general case of weighted matching in arbitrary graphs.
\end{itemize}

Because we recurse on $G[S\cup L]$ and $G[S\cup R]$, we us a separator hierarchy in which these two subgraphs (rather than $L$ and $R$) are the ones that are recursively subdivided.
In more detail, our algorithm performs the following steps:

\begin{enumerate}
\item Recursively compute a minimum-weight maximum-cardinality matching, and its corresponding system of dual variables, for the two subgraphs $G[S\cup L]$ and $G[S\cup R]$. Recall that the dual variables are numbers associated with each vertex such that, for each unmatched edge, the two numbers at its endpoints sum to at most the edge length, and for each matched edge they sum to exactly the edge length.
\item Combine the two subproblems to give a single matching (with fewer than the maximum number of matched edges) and valid system of dual variables for the whole graph $G$. To do so, for each vertex $s$ in $S$, we choose the dual variable for $s$ in~$G$ to be the minimum of the two values computed for $s$ in $G[L\cup S]$ and in $G[R\cup S]$, and we remove the matched edge incident to this vertex in one of the two subproblems, the subproblem that did not supply the minimum dual variable value. (If the two subproblems have equal dual variables at $s$, we choose arbitrarily which of the two matched edges to remove.)
\item Repeat the following steps until no more alternating paths can be found:
\begin{itemize}
\item Modify the original weight of each edge by subtracting the dual variables at its endpoints.
\item Add an artificial source vertex adjacent by zero-length edges to all unmatched vertices on one side of the bipartition. Orient the unmatched edges of the graph from the source side to the opposite side, and the matched edges in the other direction.
\item Use the algorithm of Henzinger et al. starting from this source vertex
to find the  alternating path of minimum modified weight between two unmatched vertices.
\item Augment the matching using the alternating path, trading matched and unmatched edges along the path.
\item Use the distances from the source vertex found by the path search to update the dual variables so that they remain valid.
\end{itemize}
\end{enumerate}

To update the dual variables, decrease each dual variable on the source side of the bipartition by its distance from the source, and increase each dual variable on the opposite side by its distance from the source. The sum of dual variables at the endpoints of each matched edge remains unchanged, the sum of variables on any shortest path edge increases to equal the edge length, and the sum of variables on any other edge remains at most equal to the edge length (else that edge would have supplied a shorter path).

The result of this algorithm is a maximum-cardinality matching on $G$, and a feasible system of dual variables, whose existence ensures that the matching has minimum weight (as any alternating cycle would have non-negative modified weight and therefore non-negative total cost).
This proves the following result:

\begin{lemma}
\label{lem:sparse-bipartite-matching}
Let $G$ be a bipartite graph given together with a separator hierarchy in which each subgraph of $p$ vertices has a separator of size $O(p^c)$, for some constant~$c$ with $0<c<1$. Then one can find both a minimum weight perfect matching of $G$, and the dual variables of the matching, in time $O(n^{1+c})$.
\end{lemma}

\begin{proof}
The union of the matchings in $G[S\cup L]$ and $G[S\cup R]$ is at least as large as the single maximum matching in~$G$, so the reconciliation process (in which we remove at most $|S|$ edges from the union of the matchings) produces a single matching that differs in cardinality from the maximum by at most $|S|$. Therefore, the inner loop of the algorithm involves at most $|S|$ linear-time path searches. The time analysis of \autoref{lem:sparse-bipartite-matching} follows straightforwardly as a divide-and-conquer recurrence.
\end{proof}

\subsection{Low-ply neighborhood systems}

To show that the nearest-neighborhood graph has small separators, and hence that we can quickly solve weighted matching problems on it, we need some definitions from Teng~\cite{Ten-PhD-91}.

\begin{definition}
A \emph{$k$-neighborhood system} is a system of balls such that each ball contains at most $k$ centers of balls in the system.
\end{definition}

Clearly, the nearest-neighbor balls from which $\NNG$ is defined form a $1$-neighborhood system, because each contains only its own center.

\begin{definition}
The \emph{ply} of a system of balls is the maximum number of balls that have a common point of intersection. The \emph{kissing number} $\tau_d$ of $d$-dimensional Euclidean space is the maximum number of unit-radius balls that can touch a single central unit ball without overlapping each other.
\end{definition}

The kissing number is known to depend singly-exponentially on the dimension, but for us it is more important that it is constant for any constant dimension.

\begin{lemma}[Teng~\cite{Ten-PhD-91}, Lemma~3.2]
\label{lem:ply}
Let $k$ and $d$ be fixed constants. Then
every $k$-neighborhood system in $d$-dimensional Euclidean space has ply at most $\tau_d k=O(1)$.
\end{lemma}

The following definition is not from Teng (who uses different terminology) but is standard in graph theory.

\begin{definition}
The \emph{degeneracy} of a graph $G$ is the smallest value $d$ such that every subgraph of $G$ contains a vertex of degree at most $d$.
\end{definition}

\begin{lemma}[Teng~\cite{Ten-PhD-91}, Theorem 4.2]
For any $k$-neighborhood system in a Euclidean space of bounded dimension,
the intersection graph of the balls has bounded degeneracy.
\end{lemma} 

In particular, this implies that $\NNG$ has $O(n)$ edges.
This immediately gives a speedup from $O(n^3)$ to $O(n^2\log n)$, by applying the $O(mn+n^2\log n)$ time bound for matching on the sparse bipartite graph $2\NNG$ and using the degeneracy bound to substitute $m\le n\Delta=O(n)$. But we can do better using more graph structure, a separator theorem for overlap graphs.

\begin{lemma}[Teng~\cite{Ten-PhD-91}, Theorem 5.3]
For any system of balls of bounded ply in a Euclidean space of bounded dimension $d$,
the intersection graph of the balls has separators of size $O(n^{1-1/d})$
\end{lemma} 

The following deterministic algorithmic version of this separator theorem is from~\cite{EppMilTen-FI-95}:

\begin{lemma}[Eppstein et al.~\cite{EppMilTen-FI-95}, Theorems~5.1 and~5.8]
\label{lem:deterministic-separator}
For any system of balls of bounded ply in a Euclidean space of bounded dimension $d$,
the intersection graph of the balls has separators of size $O(n^{1-1/d})$ that can be found deterministically in linear time from the balls.
By recursively constructing a hierarchy of separators, the overlap graph can be constructed from the balls in time $O(n\log n)$.
\end{lemma}

Another way of expressing the existence of small separators (of size a fractional power of $n$) is that overlap graphs are graphs of \emph{polynomial expansion}~\cite{DvoNor-SSS-15,NesOss-SGSA-12}. We remark that any hierarchy of separators for $\NNG$ can be transformed into a hierarchy of separators for $2\NNG$, simply by including both copies of each vertex of a separator for $\NNG$ in the corresponding separator for $2\NNG$.

\subsection{Convex distance functions}

A \emph{convex distance function} may be defined from any centrally symmetric convex body $B$ on $d$-dimensional Euclidean space by defining the distance between any two points $p$ and $q$ to be the infimum of the positive numbers $s$ such that a copy of $B$, centered at $p$ and expanded by a factor of $s$, contains $q$. The distance defined in this way obeys the axioms of a metric space, and gives the Euclidean space the structure of a normed vector space. Conversely, the distances in any bounded dimensional normed vector space can be defined in this way from the unit ball of the norm. Much, but not all, of the same theory for Euclidean balls goes through in these spaces. In particular, we have:

\begin{lemma}[Teng~\cite{Ten-PhD-91}, Lemma 7.2]
Let $k$ be a fixed constant, fix a convex distance function on a normed metric space of fixed dimension $d$, and let $\tau$ be the kissing number for that space (the maximum number of non-overlapping unit balls that can simultaneously touch a single central unit ball).
Then
Every $k$-neighborhood system in the space has ply at most $\tau_d k\le (3^d-1)k=O(1)$.
\end{lemma}

\begin{lemma}[Teng~\cite{Ten-PhD-91}, Lemma 7.4]
Let $k$ be a fixed constant, and fix a convex distance function on a normed metric space of fixed dimension $d$. Then the intersection graph of any $k$-neighborhood system has bounded degeneracy.
\end{lemma}

\begin{lemma}[Teng~\cite{Ten-PhD-91}, Theorem 7.1]
\label{lem:cdf-separator}
Let $k$ be a fixed constant, and fix a convex distance function on a normed metric space of fixed dimension $d$. Then the intersection graph of any $k$-neighborhood system has separators of size $O(n^{1-1/d})$ that can be found by a randomized algorithm in linear expected time from the balls.
By recursively constructing a hierarchy of separators, the overlap graph can be constructed from the balls in randomized expected time $O(n\log n)$.
\end{lemma}

\subsection{Optimal disks for low-dimensional Euclidean spaces}

Putting the results of this section together, we have our main theorem:

\begin{theorem}
Given $n$ points in a Euclidean space of constant dimension~$d\ge 2$, we can construct a system of balls centered at the given points, with maximum sum of radii, in time $O(n^{2-1/d})$.
For points in a metric space of dimension~$d$ defined by a convex distance function, we can solve the same problem in randomized expected time $O(n^{2-1/d})$.
\end{theorem}

\begin{proof}
We first find the nearest neighbors of all points.
For this, the deterministic algorithm of Vaidya~\cite{Vai-DCG-89} takes time $O(n\log n)$.
It is described by Vaidya for $L_p$ metrics, but depends only on the property that, if the space is partitioned into a grid of congruent cubes, then only a bounded number of cubes can intersect a ball whose radius is the diameter of a single cube. This is true for any convex distance function, so Vaidya's algorithm can be used to find all nearest neighbors for any convex distance function.

We then construct graph $\NNG$ and its separator hierarchy by \autoref{lem:deterministic-separator} (for Euclidean distances) or \autoref{lem:cdf-separator} (for convex distance functions).
We construct the bipartite double cover $2\NNG$ and a system of separators for it that are formed by doubling the separators in $\NNG$. Using these separators, we apply the matching algorithm of \autoref{lem:sparse-bipartite-matching} to $2\NNG$, giving two dual variables per input point. Finally, we take the average of these two dual variables to give a single radius for each point.
\end{proof}

\section{Lower-bounding the radii}
\label{sec:constrained}

As discussed in \autoref{sec:label}, in the map-labeling application of our problem it is helpful to constrain the radii of the disks or balls to all be at least $\delta$, for some given value $\delta>0$ (small enough that all $\delta$-balls are non-overlapping), so that we avoid zero-radius or overly-small balls in our solution.

This constrained problem, for any given value of $\delta$, can be transformed to an instance of our unconstrained radius-sum optimization problem by modifying the input distances, as follows.

\begin{itemize}
\item Form a collection of balls of radius $\delta$ around each given point $p_i$, and let $d^\circ(p_i,p_j)=d(p_i,p_j)-2\delta$ be the distance between each pair of balls. Then $d^\circ$ is non-negative and symmetric, but it might not obey the triangle inequality, which we will need in our algorithms (in particular, in \autoref{lem:dual-equality}).
\item Let $d^*$ be the shortest path distance in a complete geometric graph whose edge weights are $d^\circ$. The distance $d^*$ defined in this way is non-negative, symmetric, and obeys the triangle inequality.
\item Find a system of non-negative radii $r_i^*$ of disjoint metric balls for $d^*$, maximizing the sum of radii.
\item Set $r_i=r_i^*+\delta$.
\end{itemize}

\begin{theorem}
For $r_i$ computed by the steps outlined above,
the balls of radius $r_i$ around each point $p_i$ of the original point set all have radius at least $\delta$, and have the maximum total radius of any system of balls of radius at least $\delta$.
\end{theorem}

\begin{proof}
Let $r_i^\circ$ be a system of radii, obeying the constraints that for each pair of points $p_i$ and $p_j$, $r_i^\circ+r_j^\circ\le d^\circ(p_i,p_j)$ and maximizing the sum of radii.
Then the constraint that $r_i^\circ+r_j^\circ\le d^\circ(p_i,p_j)$ is equivalent (by simple algebraic manipulation) to the constraint that
\[
(r_i^\circ+\delta)+(r_j^\circ+\delta)\ge d(p_i,p_j),
\]
or in words that the two ball radii derived from $r_i^\circ$ and $r_j^\circ$ form non-overlapping balls around the original points. Clearly, in addition, these radii are at least $\delta$ and are otherwise unconstrained. So if we could find the system of radii $r_i^\circ$, the radii $r_i^\circ+\delta$ derived from them would solve the problem stated in the theorem. It remains to show that each radius $r_i^*$ found by the procedure outlined above actually equals $r_i^\circ$.

To see this, consider the transformation of distances made in going from $d^\circ$ to $d^*$.
This transformation can only decrease any pairwise distance; that is, for every $i$ and $j$,
\[
d^*(p_i,p_j)\le d^\circ(p_i,p_j).
\]
For, the distance $d^\circ(p_i,p_j)$ is the length of the edge $p_ip_j$ in the complete weighted graph used to compute $d^*$, and this edge forms one of the paths in the shortest path distance computation. Because $d^*$ has distances that are no greater than $d^\circ$, the radii $r^*$ must also be no greater than the corresponding values of $r^\circ$.

To complete the proof, we show that the system of radii $r^\circ$ automatically obey the constraints $r_i^\circ+r_j^\circ\le d^*(p_i,p_j)$ used to determine the radii $r^*$. Since both $r^\circ$ and $r^*$ are defined by optimization problems with the same objective function but with tighter constraints on $r^*$, and since these tighter constraints are already obeyed by $r^\circ$, it will follow that the optimization problems defining $r^\circ$ and $r^*$ have identical solution sets.
Therefore, consider any two points $p_i$ and $p_j$, and their radii $r_i^\circ$ and $r_j^\circ$; we must show that $r_i^\circ+r_j^\circ\le d^*(p_i,p_j)$. To see this, consider the shortest path from $p_i$ to $p_j$ in the complete weighted graph from which we defined $d^*$, and let $S$ be the sum, over all the edges in this path, of the radii $r^\circ$ at the two endpoints of this edge. Then each edge's contribution to the sum is less than its length, for otherwise the two circles of radius $r^\circ$ at its endpoints would violate the constraint that their sum is at most $d^\circ$. Therefore, $S$ is less than the total length of the path, which is $d^*(p_i,p_j)$. But both $r_i^\circ$ and $r_j^\circ$ contribute to $S$ in the terms for the first and last edge to the path, and all the other contributions are non-negative. Therefore, $r_i^\circ+r_j^\circ\le S\le d^*(pi,p_j)$, as desired.
\end{proof}

We leave open for future research the problem of combining this technique for lower-bounding the radii of systems of balls with our speedups for Euclidean radius-sum optimization.

{\small\raggedright
\bibliographystyle{abuser}
\bibliography{radius-sum}}

\begin{thebibliography}{10}
\urlstyle{rm}

\bibitem{AltArkBro-SoCG-06}
H.~Alt, E.~M. Arkin, H.~Br{\"o}nnimann, J.~Erickson, S.~P. Fekete, C.~Knauer,
  J.~Lenchner, J.~S.~B. Mitchell, and K.~Whittlesey.
\newblock {Minimum-cost coverage of point sets by disks}.
\newblock {\em Proc. 22nd Symp. Computational Geometry (SoCG 2006)},
  pp.~449{--}458, 2006, \href{http://dx.doi.org/10.1145/1137856.1137922}%
{doi:\nolinkurl{10.1145/1137856.1137922}},
  \href{https://www.ams.org/mathscinet-getitem?mr=2389355}%
{MR2389355}.

\bibitem{ChaFekHof-WADS-11}
E.~W. Chambers, S.~P. Fekete, H.-F. Hoffmann, D.~Marinakis, J.~S.~B. Mitchell,
  V.~Srinivasan, U.~Stege, and S.~Whitesides.
\newblock {Connecting a set of circles with minimum sum of radii}.
\newblock {\em Proc. 12th Int. Symp. Algorithms and Data Structures (WADS
  2011)}, pp.~183{--}194. Springer, LNCS 6844, 2011,
  \href{http://dx.doi.org/10.1007/978-3-642-22300-6_16}%
{doi:\nolinkurl{10.1007/978-3-642-22300-6_16}},
  \href{https://www.ams.org/mathscinet-getitem?mr=2863135}%
{MR2863135}.

\bibitem{Cla-FOCS-83}
K.~L. Clarkson.
\newblock {Fast algorithms for the all nearest neighbors problem}.
\newblock {\em Proc. 24th IEEE Symp. Foundations of Computer Science (FOCS
  '83)}, pp.~226{--}232, 1983, \href{http://dx.doi.org/10.1109/SFCS.1983.16}%
{doi:\nolinkurl{10.1109/SFCS.1983.16}}.

\bibitem{CohMeg-SICOMP-94}
E.~Cohen and N.~Megiddo.
\newblock {Improved algorithms for linear inequalities with two variables per
  inequality}.
\newblock {\em SIAM J. Comput.} 23(6):1313{--}1347, 1994,
  \href{http://dx.doi.org/10.1137/S0097539791256325}%
{doi:\nolinkurl{10.1137/S0097539791256325}},
  \href{https://www.ams.org/mathscinet-getitem?mr=1303338}%
{MR1303338}.

\bibitem{CooCunPul-CO-98}
W.~J. Cook, W.~H. Cunningham, W.~R. Pulleyblank, and A.~Schrijver.
\newblock {\em {Combinatorial optimization}}.
\newblock Wiley-Interscience Series in Discrete Mathematics and Optimization.
  John Wiley {\&} Sons, 1998,
  \href{https://www.ams.org/mathscinet-getitem?mr=1490579}%
{MR1490579}.

\bibitem{DodMarMir-SODA-97}
S.~Doddi, M.~V. Marathe, A.~Mirzaian, B.~M.~E. Moret, and B.~Zhu.
\newblock {Map labeling and its generalizations}.
\newblock {\em Proc. 8th ACM-SIAM Symp. on Discrete Algorithms (SODA '97)},
  pp.~148{--}157, 1997,
  \href{https://www.ams.org/mathscinet-getitem?mr=1447660}%
{MR1447660}.

\bibitem{DuaSu-SODA-12}
R.~Duan and H.-H. Su.
\newblock {A scaling algorithm for maximum weight matching in bipartite
  graphs}.
\newblock {\em Proc. 23rd ACM-SIAM Symp. on Discrete Algorithms (SODA '12)},
  pp.~1413{--}1424, 2012,
  \href{https://www.ams.org/mathscinet-getitem?mr=3205301}%
{MR3205301}.

\bibitem{DvoNor-SSS-15}
Z.~Dvo{\v{r}}{\'a}k and S.~Norin.
\newblock {Strongly sublinear separators and polynomial expansion}.
\newblock Electronic preprint arxiv:1504.04821, 2015.

\bibitem{EppMilTen-FI-95}
D.~Eppstein, G.~L. Miller, and S.-H. Teng.
\newblock {A deterministic linear time algorithm for geometric separators and
  its applications}.
\newblock {\em Fundamenta Informaticae} 22(4):309{--}331, 1995,
  \href{https://www.ams.org/mathscinet-getitem?mr=1360950}%
{MR1360950}.

\bibitem{EppWor-WADS-09}
D.~Eppstein and K.~A. Wortman.
\newblock {Optimal embedding into star metrics}.
\newblock {\em Proc. Algorithms and Data Structures Symposium (WADS 2009)},
  pp.~290{--}301. Springer-Verlag, Lecture Notes in Computer Science 5664,
  2009, \href{http://dx.doi.org/10.1007/978-3-642-03367-4_26}%
{doi:\nolinkurl{10.1007/978-3-642-03367-4_26}},
  \href{http://arxiv.org/abs/0905.0283}{arXiv:0905.0283},
  \href{https://www.ams.org/mathscinet-getitem?mr=2550615}%
{MR2550615}.
\newblock Winner, best paper award.

\bibitem{FreTar-JACM-87}
M.~L. Fredman and R.~E. Tarjan.
\newblock {Fibonacci heaps and their uses in improved network optimization
  algorithms}.
\newblock {\em J. ACM} 34(3):596{--}615, 1987,
  \href{http://dx.doi.org/10.1145/28869.28874}%
{doi:\nolinkurl{10.1145/28869.28874}},
  \href{https://www.ams.org/mathscinet-getitem?mr=904195}%
{MR904195}.

\bibitem{GibKanKro-Algo-10}
M.~Gibson, G.~Kanade, E.~Krohn, I.~A. Pirwani, and K.~Varadarajan.
\newblock {On metric clustering to minimize the sum of radii}.
\newblock {\em Algorithmica} 57(3):484{--}498, 2010,
  \href{http://dx.doi.org/10.1007/s00453-009-9282-7}%
{doi:\nolinkurl{10.1007/s00453-009-9282-7}},
  \href{https://www.ams.org/mathscinet-getitem?mr=2609050}%
{MR2609050}.

\bibitem{HenKleRao-JCSS-97}
M.~R. Henzinger, P.~Klein, S.~Rao, and S.~Subramanian.
\newblock {Faster shortest-path algorithms for planar graphs}.
\newblock {\em J. Comput. System Sci.} 55(1):3{--}23, 1997,
  \href{http://dx.doi.org/10.1006/jcss.1997.1493}%
{doi:\nolinkurl{10.1006/jcss.1997.1493}},
  \href{https://www.ams.org/mathscinet-getitem?mr=1473046}%
{MR1473046}.

\bibitem{HocNao-SICOMP-94}
D.~S. Hochbaum and J.~Naor.
\newblock {Simple and fast algorithms for linear and integer programs with two
  variables per inequality}.
\newblock {\em SIAM J. Comput.} 23(6):1179{--}1192, 1994,
  \href{http://dx.doi.org/10.1137/S0097539793251876}%
{doi:\nolinkurl{10.1137/S0097539793251876}},
  \href{https://www.ams.org/mathscinet-getitem?mr=1303329}%
{MR1303329}.

\bibitem{IndMat-HDCG-04}
P.~Indyk and J.~Matou{\v{s}}ek.
\newblock {Low-distortion embeddings of finite metric spaces}.
\newblock {\em Handbook of Discrete and Computational Geometry}, 2nd edition,
  pp.~177{--}196. CRC, Discrete Mathematics and its Applications (Boca Raton),
  2004, \href{http://dx.doi.org/10.1201/9781420035315.ch8}%
{doi:\nolinkurl{10.1201/9781420035315.ch8}}.

\bibitem{Isb-CMH-64}
J.~R. Isbell.
\newblock {Six theorems about injective metric spaces}.
\newblock {\em Comment. Math. Helv.} 39:65{--}76, 1964,
  \href{http://dx.doi.org/10.1007/BF02566944}%
{doi:\nolinkurl{10.1007/BF02566944}}.

\bibitem{JiaBerQin-ISAAC-04}
M.~Jiang, S.~Bereg, Z.~Qin, and B.~Zhu.
\newblock {New bounds on map labeling with circular labels}.
\newblock {\em Proc. 15th Int. Symp. Algorithms and Computation (ISAAC 2004)},
  pp.~606{--}617. Springer, LNCS 3341, 2004,
  \href{http://dx.doi.org/10.1007/978-3-540-30551-4_53}%
{doi:\nolinkurl{10.1007/978-3-540-30551-4_53}},
  \href{https://www.ams.org/mathscinet-getitem?mr=2158365}%
{MR2158365}.

\bibitem{JunPul-Algo-95}
M.~J{\"u}nger and W.~Pulleyblank.
\newblock {New primal and dual matching heuristics}.
\newblock {\em Algorithmica} 13(4):357{--}380, 1995,
  \href{http://dx.doi.org/10.1007/BF01293485}%
{doi:\nolinkurl{10.1007/BF01293485}},
  \href{https://www.ams.org/mathscinet-getitem?mr=1318310}%
{MR1318310}.

\bibitem{LipTar-SICOMP-80}
R.~J. Lipton and R.~E. Tarjan.
\newblock {Applications of a planar separator theorem}.
\newblock {\em SIAM J. Comput.} 9(3), 1980,
  \href{http://dx.doi.org/10.1137/0209046}%
{doi:\nolinkurl{10.1137/0209046}},
  \href{https://www.ams.org/mathscinet-getitem?mr=584516}%
{MR584516}.

\bibitem{NesOss-SGSA-12}
J.~Ne{\v{s}}et{\v{r}}il and P.~Ossona~de Mendez.
\newblock {\em {Sparsity: Graphs, Structures, and Algorithms}}.
\newblock Algorithms and Combinatorics~28. Springer, 2012,
  \href{http://dx.doi.org/10.1007/978-3-642-27875-4}%
{doi:\nolinkurl{10.1007/978-3-642-27875-4}},
  \href{https://www.ams.org/mathscinet-getitem?mr=2920058}%
{MR2920058}.

\bibitem{Str-IJCGA-01}
T.~Strijk and A.~Wolff.
\newblock {Labeling points with circles}.
\newblock {\em Int. J. Comput. Geom. Appl.} 11(02):181{--}195, 2001,
  \href{http://dx.doi.org/10.1142/s0218195901000444}%
{doi:\nolinkurl{10.1142/s0218195901000444}},
  \href{https://www.ams.org/mathscinet-getitem?mr=1831031}%
{MR1831031}.

\bibitem{Tar-DSNA-83}
R.~E. Tarjan.
\newblock {\em {Data Structures and Network Algorithms}}.
\newblock CBMS-NSF Regional Conference Series in Applied Mathematics~44. SIAM,
  1983, \href{http://dx.doi.org/10.1137/1.9781611970265}%
{doi:\nolinkurl{10.1137/1.9781611970265}},
  \href{https://www.ams.org/mathscinet-getitem?mr=826534}%
{MR826534}.

\bibitem{Ten-PhD-91}
S.-H. Teng.
\newblock {\em {Points, spheres, and separators: a unified approach to graph
  partitioning}}.
\newblock Ph.D. thesis, Carnegie-Mellon Univ., School of Computer Science,
  1991, \href{https://www.ams.org/mathscinet-getitem?mr=2687483}%
{MR2687483}.

\bibitem{Vai-DCG-89}
P.~M. Vaidya.
\newblock {An $O(n\log n)$ algorithm for the all-nearest-neighbors problem}.
\newblock {\em Discrete Comput. Geom.} 4(2):101{--}115, 1989,
  \href{http://dx.doi.org/10.1007/BF02187718}%
{doi:\nolinkurl{10.1007/BF02187718}},
  \href{https://www.ams.org/mathscinet-getitem?mr=973540}%
{MR973540}.

\bibitem{Vaz-AA-01}
V.~V. Vazirani.
\newblock {\em {Approximation Algorithms}}.
\newblock Springer, 2001, \href{http://dx.doi.org/10.1007/978-3-662-04565-7}%
{doi:\nolinkurl{10.1007/978-3-662-04565-7}},
  \href{https://www.ams.org/mathscinet-getitem?mr=1851303}%
{MR1851303}.

\bibitem{WagWol-CGTA-97}
F.~Wagner and A.~Wolff.
\newblock {A practical map labeling algorithm}.
\newblock {\em Comp. Geom. Th. {\&} Appl.} 7(5-6):387{--}404, 1997,
  \href{http://dx.doi.org/10.1016/s0925-7721(96)00007-7}%
{doi:\nolinkurl{10.1016/s0925-7721(96)00007-7}},
  \href{https://www.ams.org/mathscinet-getitem?mr=1447248}%
{MR1447248}.

\end{thebibliography}

\end{document}